\documentclass[prx,twocolumn,english,superscriptaddress,floatfix,longbibliography,tikz]{revtex4-2}

\usepackage[utf8]{inputenc}
\usepackage[english]{babel}
\usepackage{listings}

\usepackage{xcolor}
\usepackage{physics}
\usepackage{graphicx}
\usepackage{amsmath}
\usepackage{amsthm}
\usepackage[section]{algorithm}
\usepackage{algpseudocode}
\usepackage{soul}
\usepackage{dsfont}
\usepackage{amsfonts}

\theoremstyle{definition}
\newtheorem{definition}{Definition}
\newtheorem{theorem}{Theorem}
\newtheorem{lemma}{Lemma}[theorem]

\addto\captionsenglish{}

\newcommand{\appsection}[1]{\section{\MakeUppercase{#1}}}

\makeatletter
\def\bibsection{%
   \par
   \begingroup
    \baselineskip26\p@
    \bib@device{\hsize}{72\p@}%
   \endgroup
   \nobreak\@nobreaktrue
   \addvspace{19\p@}%
  }%
\makeatother

\begin{document}

\preprint{ }

\title{Automated detection of symmetry-protected subspaces in quantum simulations}
\author{Caleb Rotello}
        \email{caleb.rotello@nrel.gov}
        \affiliation{Department of Physics and Quantum Engineering Program, Colorado School of Mines, Golden CO, 80401}
        \affiliation{National Renewable Energy Laboratory, Golden CO, 80401}
\author{Eric B. Jones}
        \email{eric.jones@coldquanta.com}
        \email{eric.jones@infleqtion.com}
        \affiliation{Infleqtion, Louisville CO, 80027}
        \affiliation{National Renewable Energy Laboratory, Golden CO, 80401}
\author{Peter Graf}
        \email{Peter.Graf@nrel.gov}
        \affiliation{National Renewable Energy Laboratory, Golden CO, 80401}
\author{Eliot Kapit}
        \email{ekapit@mines.edu}
        \affiliation{Department of Physics and Quantum Engineering Program, Colorado School of Mines, Golden CO, 80401}

\date{\today}

\begin{abstract}
The analysis of symmetry in quantum systems is of utmost theoretical importance, useful in a variety of applications and experimental settings, and is difficult to accomplish in general. Symmetries imply conservation laws, which partition Hilbert space into invariant subspaces of the time-evolution operator, each of which is demarcated according to its conserved quantity. We show that, starting from a chosen basis, any invariant, symmetry-protected subspaces which are diagonal in that basis are discoverable using transitive closure on graphs representing state-to-state transitions under $k$-local unitary operations. Importantly, the discovery of these subspaces relies neither upon the explicit identification of a symmetry operator or its eigenvalues nor upon the construction of matrices of the full Hilbert space dimension. We introduce two classical algorithms, which efficiently compute and elucidate features of these subspaces. The first algorithm explores the entire symmetry-protected subspace of an initial state in time complexity linear to the size of the subspace by closing local basis state-to-basis state transitions. The second algorithm determines, with bounded error, if a given measurement outcome of a dynamically-generated state is within the symmetry-protected subspace of the state in which the dynamical system is initialized. We demonstrate the applicability of these algorithms by performing post-selection on data generated from emulated noisy quantum simulations of three different dynamical systems: the Heisenberg-XXX model and the $T_6$ and $F_4$ quantum cellular automata. Due to their efficient computability and indifference to identifying the underlying symmetry, these algorithms lend themselves to the post-selection of quantum computer data, optimized classical simulation of quantum systems, and the discovery of previously hidden symmetries in quantum mechanical systems.
        
\end{abstract}

\maketitle

\section{Introduction}\label{sec:introduction}

The analysis of symmetry is a central tool in physics and has enabled some of the most profound discoveries in the field. Noether's Theorem famously connects the symmetries of a system's action with conservation laws to which that system's equations of motion are subject~\cite{Noether1918}. Generally, the analysis of symmetry, or the breaking thereof, allows one to constrain theories~\cite{gaillard1999standard}, solve equations of motion more efficiently~\cite{kozlov1983integrability}, and identify phases of matter~\cite{Landau:1937obd}. Applications of symmetry analysis in quantum information include, but are not limited to: quantum error correction~\cite{Fowler_2012}, error mitigation on quantum hardware~\cite{PhysRevA.98.062339, PhysRevA.100.010302, PhysRevA.95.042308, Cai_2021, McClean2020}, and quantum machine learning model design~\cite{PRXQuantum.3.030341}.

Quantum computing can efficiently simulate quantum dynamics in regimes where classical simulation becomes impossible~\cite{Daley2022-vl}. However, current quantum processors operate in a regime severely constrained by noise, with error rates not yet sufficiently below most error correction thresholds~\cite{acharya2022suppressing}. Error mitigation will therefore be
critical in the interim before fault-tolerant architectures can be scaled~\cite{cai2022quantum, quek2022exponentially, doi:10.1021/acs.chemrev.8b00803}. In recent work, and despite its limitations~\cite{Takagi_2022}, the technique of
\textit{post-selection} has proven useful to mitigate errors and extract useful results from quantum simulation experiments (see e.g.,~\cite{arute2020observation, doi:10.1126/science.abb9811, Jones2022}). 
Post-selection works by identifying measured states that could only have come from error processes and excluding them from the statistics used to calculate
output quantities. The most obvious example is a conserved quantity such as particle number.  In such a case, any measured state that does not preserve the conserved quantity must be
the result of errors. Due to the connection established by Noether, a more fundamental way to describe post-selection is with respect to symmetry.  
According to this description,
post-selection works by checking the eigenvalue of the simulation's ``fiducial" (i.e., initial) state under the symmetry operator 
against the corresponding symmetry operator eigenvalues (e.g., value of a conserved quantity) of individual measurement results in the dynamically-generated output state of the simulation. If a particular measurement outcome registers a different eigenvalue under the symmetry operator than the fiducial state does, then the measurement is a result of error and can be discarded. This procedure is restricted in the scope of its application, as symmetries of a quantum system and their corresponding operators 
(i.e., conserved quantities) are typically either engineered into the dynamics ``by hand'' or identified by clever theoretical intuition. For a generic quantum system, the relevant symmetry operator(s) may not be obvious a priori and may be difficult to identify. 
Being able to perform post-selection in a manner that does not require explicit identification of a symmetry operator would greatly increase the technique's applicability; this is the subject of this paper.  

As a corollary, such an operator-free method for error detection also enables additional applications such as more efficient classical simulation of quantum systems via computational basis state reduction. For example, particle number conservation in hardcore boson models can be used both for post-selection in quantum simulation and for reducing the basis state set size from $2^n$ to $\binom{n}{N}$, where $n$ is the number of lattice sites and $N$ the number of particles, in classical simulations~\cite{hebert2001quantum}. Interestingly, the identification of symmetry or conserved quantities in some instances can make classical simulation so efficient that it can obviate the need for quantum computation altogether~\cite{anschuetz2022efficient}. Finally, in certain special cases, one may be able to infer the explicit form of a symmetry operator by inspection of the reduced basis set.
 
In this paper, we provide algorithms to efficiently make use of symmetry in an operator-free manner. To do this, our methods create the subspace of measurement basis states which would share a conserved quantity of some commuting symmetry operator that is diagonal in that measurement basis, without needing to explicitly create that operator.  We call such a space a \textit{symmetry-protected subspace} (SPS).  In the language of linear algebra, these are \textit{invariant subspaces} of the evolution operator; they are subspaces,
determined by the initial state, from which the evolution cannot escape.  To reap the benefits of symmetry we only need to find the SPS of the initial state, not a conserved quantity, much less an explicit symmetry.  However, naively, to find an SPS we need to actually evolve the system in the full Hilbert space, which is 
exponentially large in the number of qubits (particles, spins, etc.).  Sections~\ref{sec:graphtheory} through~\ref{sec:algorithm2} describe the formulation and algorithms by which we avoid this exponential scaling, but here we provide a non-technical overview. 

First, note that most Hamiltonians and resulting Unitary evolution operators are built from a number of \textit{local} operators.  For example, the Heisenberg-XXX 
model described below consists only of nearest-neighbor interactions.  So, at some level, we have an intuition that the dynamics, thus the SPSs, should be
derivable, like the Hamiltonian itself, from a combination of local operations, and that local operations are inherently less computationally expensive to work with.  This is indeed the case, as shown below.

Next, note that to say that a wavefunction is \textit{in} a symmetry-protected or invariant subspace is to say that it is and remains throughout dynamic evolution a linear combination of basis states in that subspace and that subspace alone.  And if we care only about finding the subspace, we do not need to keep track of the actual linear combination (i.e., both the basis vectors and their amplitudes) but only the basis vectors.  This ``binarization" of the evolution is critical, because it allows us to adopt a graph-theoretic framework that is vastly more efficient for finding and searching SPSs.

This also leads to an important restriction in our work; using our methods, we can only automate the discovery of symmetry which is diagonal in a chosen basis. We work in the computational $Z$ basis throughout this paper, though extensions to other bases are of course possible by rewriting the time evolution operators in the new basis and performing the same procedure we describe below. Our automated methods should thus be viewed as a tool to find symmetry-protected subspaces within a given basis (if, of course, they exist), and to potentially improve classical and noisy quantum simulations based on those discovered subspaces. However, they still require an intelligently guessed initial basis as a starting point. For the problems we consider in this work, the computational basis is sufficient to derive novel results, though more complex choices can be required in other cases.

With this caveat in mind, once the basis is chosen we create an undirected and unweighted graph, called the state interaction graph, which describes all possible state-to-state transitions over a single application of a unitary evolution operator. The transitive closure of this graph fragments the Hilbert space of the system, represented in a particular measurement basis, into a cluster graph, whose subgraphs are each a symmetry-protected subspace. 

Our main results are two classical algorithms that efficiently construct and work within these subspaces. The first, Algorithm~\ref{alg:makesps}, uses ``transitive closure" on local operations to explore and explicitly construct the full SPS of an initial state, which enables the partition of the Hilbert space into a set of disjoint SPSs. This algorithm scales linearly in both the number of local operators from which the global operator is constructed and the size of the SPS, which is a huge improvement over the exponential scaling of the naive ``full evolution" approach.  However, because the SPS itself can be exponentially large (albeit with an asymptotically smaller prefactor), the second algorithm,  Algorithm~\ref{alg:greedypath}, finds a \textit{path} of local operations through a set of SPS graphs from an initial to final (i.e., measured) state to determine if they lie within the \textit{same} SPS (and thus the final state is valid). This algorithm scales as the number of local operations raised to a small integer power (that can be tuned for accuracy/performance) times the length of the path, thus completely eliminating any exponential scaling.

This paper is structured as follows. In Sec.~\ref{sec:preliminaries}, we define symmetry-protected subspaces and the three quantum systems that will serve as our ``exemplars" throughout. In Sec.~\ref{sec:graphtheory} we outline our novel graph theoretical approach to quantum simulations. Section~\ref{sec:algorithm1} discusses the algorithm for computing an entire symmetry-protected subspace. Section~\ref{sec:algorithm2} provides a more efficient algorithm to verify if two states exist within the same symmetry-protected subspace. Finally, we demonstrate the effectiveness of symmetry-protected subspaces to mitigate error in quantum simulations on a classical emulator in Sec.~\ref{sec:psdemo}.

\section{Preliminaries}\label{sec:preliminaries}
\subsection{Symmetry-protected subspaces}\label{subsec:symmetryprotectedsubspace}
Consider a quantum system undergoing unitary evolution according to the operator $U(t)$ for a time $t$. The operator $U(t)$ can represent continuous time evolution, but also includes other cases, such as discrete time evolution. We say that $U(t)$ is invariant under the action of an operator $S$ if $[S,U(t)]=0$ for all times $t$. In this instance, $S$ is a symmetry operator. For a basis in which $S$ is diagonal, states can be labeled by their eigenvalues under the action of $S$: $S|s,b\rangle = s|s,b\rangle$, where $b$ is some other (set of) label(s), which could represent, for example, a computational basis state integer encoding or a many-body eigenvalue. Suppose now that we initialize the dynamics in a state of definite $s$: $|\psi^s_0\rangle = \sum_{b} \alpha^s_{b} |s,b\rangle$, and evolve under $U(t)$. Given the commutativity of the symmetry and evolution operators, the action of the symmetry operator on the output state is: $S[U(t)|\psi^s_0\rangle]=U(t)S|\psi^s_0\rangle=U(t)s|\psi^s_0\rangle=s[U(t)|\psi^s_0\rangle]$, indicating that the eigenvalue $s$ is conserved under the evolution for all time. Our methods identify states that would share an eigenvalue under $S$ without 
explicitly knowing $S$. To do this, we use the notion of a symmetry-protected subspace, also known as an invariant subspace, which we define below.

\begin{definition}[Symmetry-Protected Subspace]\label{def:sps}
        Let $\mathcal{H}_d$ be a Hilbert space of dimension $d$ spanned by an orthonormal set of basis vectors, $B(\mathcal{H}_d)=\{|b\rangle\}$ . A subspace $G \subseteq \mathcal{H}_d$, which is spanned by a subset of $B(\mathcal{H}_d)$, 
        is a symmetry-protected subspace of unitary operator $U(t)$ if and only if a projection onto $G$, $P_G = \sum_{b\in G} \ket{b}\bra{b}$, obeys the commutation relation 
        $[P_G, U(t)]=0$. 
\end{definition}

Note that while our definition emphasizes the connection to symmetries, an SPS, as defined above, is indeed an invariant subspace according to the usual definition~\cite{radjavi2003invariant}, which is simply that $\forall \ket{g} \in G, \: U(t) \ket{g} \in G$, because if $\ket{g} \in G$, 
\begin{equation}
U(t) \ket{g} = U(t) P_G \ket{g} = P_G [U(t) \ket{g}] \in G,
\label{eq:invsubspace}
\end{equation}
where we have used the commutativity of $P_G$ and $U(t)$ and the fact that by definition the result of applying $P_G$
to anything is in $G$ . Another consequence of Def.~\ref{def:sps} is that if a particular basis state is not in $G$, then the transition matrix element to that state from any state in $G$ is strictly zero.

\begin{lemma}\label{lem:thm1_connection}
Let $\ket{g}\in G$ be an arbitrary element of a symmetry-protected (i.e., invariant) subspace $G$ of an evolution operator $U(t)$ and let $\ket{b} \in B(\mathcal{H}_d)$ be a basis vector outside of $G$, $\ket{b} \notin G$. Then, $\bra{b}U(t)\ket{g}=0$ for any time $t$.
\end{lemma}

\begin{proof}
Using Def.~\ref{def:sps}, $\bra{b}U(t)\ket{g}=\bra{b}U(t)P_G\ket{g}=\bra{b}P_GU(t)\ket{g}=[P_G\ket{b}]^{\dagger}U(t)\ket{g}=0$, since the projection operator onto $G$ annihilates states outside of $G$.
\end{proof}

In Sec.~\ref{sec:algorithm1} we present an algorithm that prescriptively constructs a subspace, denoted $G_{\ket{\psi_0}}$, of a particular initial state, $\ket{\psi_0}$. We also show in the corresponding theorem, Thm.~\ref{thm:spstheorem}, that subspaces so constructed share the property described in Lem.~\ref{lem:thm1_connection}, which relies on Def.~\ref{def:sps} for its proof. Hence, $G_{\ket{\psi_0}}$ constructed according to the procedure in Sec.~\ref{sec:algorithm1} are symmetry-protected subspaces.

Similar to that of the symmetry operator $S$, the commutation relation involving $P_G$ also identifies a dynamical invariance, since for any $|g\rangle \in G$, $P_G[U(t)|g\rangle] = [U(t)|g\rangle]$. That is, the state $U(t)|g\rangle$ is an eigenvector of the projector with eigenvalue $1$. Therefore, while the projector $P_G$ does not identify conserved symmetry eigenvalues, it does indicate when such an eigenvalue exists. In Sec.~\ref{sec:graphtheory}, we will show how discovering $P_G$, rather than $S$ directly, empowers our algorithms to discover underlying invariances.

\subsection{Quantum simulations}\label{subsec:quantumsimulation}

The general form for unitary evolution operators we assume is

\begin{equation}\label{eqn:timeevol}
U(t) = \mathcal{O}_{op} \Bigg[ \prod_{j=1}^p \prod_{i=1}^m U_i(\tau_j)\Bigg],
\end{equation}
where $t$ is the total duration of evolution, $\mathcal{O}_{op}$ denotes some operator ordering, such as time ordering, and $p$ is the number of time-steps used to evolve to $t$. Each $U_i(\tau_j)$ acts locally on $k$-qubits and is parameterized by a real time coordinate $\tau_j$. Evolution unitaries of the form in Eq.~\eqref{eqn:timeevol} can evolve over discrete time, where each $\tau_j$ is finite, or discretized-continuous time, where each $\tau_j$ is (ideally infinitessimally) small in order to minimize Trotter error.

In the instance where a $k$-local Hamiltonian, $H = \sum_i h_i$, is known, Eq.~\eqref{eqn:timeevol} results from the local dynamics governed by the $h_i$ via the Trotter-Suzuki formula~\cite{Suzuki1991}, and $p$ corresponds to the number of Trotter steps. We refer to the operator $U$ as the relevant quantum system and examine the dynamics and associated symmetry-protected subspaces of three exemplary systems.

\subsubsection{Heisenberg-XXX}\label{subsubsec:heisenbergXXX}
The one-dimensional Heisenberg-XXX model for $n$ spin-$1/2$ particles with nearest-neighbor interactions is given by the following Hamiltonian:
\begin{equation}\label{eq:xxx_ham}
        H = \sum_{i=0}^{n-2} X_{i}X_{i+1} + Y_{i}Y_{i+1} + Z_{i}Z_{i+1},
\end{equation}
where $X_i, Y_i,$ and $Z_i$ are Pauli operator acting on spin $i$. The model conserves total spin in the $Z$ basis, represented by the operator
\begin{equation}
        S^z = \sum_{i=0}^{n-1} Z_{i},
\end{equation}
as well as the correspondingly-defined operators $S^x$ and $S^y$. Quantum simulation of the Heisenberg model on a digital quantum processor can be achieved via exponentiation and Trotterization of Eq.~\eqref{eq:xxx_ham}. In such quantum simulation experiments, one usually picks one or a number of qubit bases in which to measure. The symmetry operators $S^{x, y, z}$ can be used to mitigate errors in the all-qubit $X,Y,Z$ measurement bases, respectively, via post-selection. In the context of classical simulation of the XXX model, the symmetries can be used to constrain the number of basis states included in the dynamics.

\subsubsection{$T_6$ quantum cellular automata}\label{subsubsec:t6qca}
The one-dimensional $T_6$ quantum cellular automata (QCA) rule has recently come to interest within the context of quantum complexity science as a dynamical small-world mutual information network generator~\cite{Jones2022} and a QCA Goldilocks rule~\cite{Hillberry2021}. Its discrete-time unitary update can be derived from a parent Hamiltonian, but it is more natural to define the system by specifying the simulation unitary for a discrete time $t=p$, directly:
\begin{equation}
\begin{split}
U(T_6;t)&=\mathcal{O}_{op}^{time} \Bigg[\prod_{j=1}^{p} \prod_{i=3, 5, \ldots}^{n-1} U_i(\tau_j) \prod_{i=2,4, \ldots}^{n-1}U_i(\tau_j)\Bigg]\\
U_i(\tau_j) &= \sum_{\alpha,\beta=0}^1 P_{i-1}^{(\alpha)} \otimes (H_i)^{\delta_{\alpha+\beta,1}} \otimes P_{i+1}^{(\beta)},
\end{split}
\end{equation}
where $P^{(\alpha)}_i = \ket{\alpha_i}\bra{\alpha_i}$ for $\alpha=0,1$ is the projection operator onto the corresponding state of qubit $i$, $H_i$ is the Hadamard operator, and $\delta_{\alpha+\beta,1}$ is the Kronecker delta function. At each time step, a Hadamard is applied to a qubit only if exactly one of its neighbors is in the $\ket{1}$ state (i.e., $\alpha + \beta = 1$) and does nothing otherwise. It has a known $Z$ basis symmetry related to domain-wall conservation:
\begin{equation} 
        S = \sum_{i=0}^{n} Z_{i}Z_{i+1},
\end{equation} 
 where it should be understood that indices $i=1, \ldots, n$ refer to dynamical, computational qubits while the indices $i=0$ and $i=n+1$ refer to non-dynamical qubits fixed to the $\ket{0}$ state.

\subsubsection{$F_4$ quantum cellular automata}\label{subsubsec:f4qca}

The one-dimensional $F_4$ QCA with nearest- and next-nearest-neighbor connectivity is another Goldilocks rule~\cite{Hillberry2021}. It is also most easily specified by its simulation unitary for discrete-time duration $t=p$:
\begin{equation}
\begin{split}
U(F_4;t) = \mathcal{O}_{op}^{time} \Bigg[ \prod_{j=1}^p U(\tau_j) \Bigg],
\end{split}
\end{equation}
where if the time step index, $j$, is even
\begin{equation}
    U(\tau_{j=\text{even}}) = \prod_{i = 2, 5, 8, \ldots} U_i(\tau_j) \prod_{i=3,6,9,\ldots} U_i(\tau_j) \prod_{i=4,7,10,\ldots} U_i(\tau_j)
\end{equation}
and if $j$ is odd, then
\begin{equation}
    U(\tau_{j=\text{odd}}) = \prod_{i = 3,6,9, \ldots} U_i(\tau_j) \prod_{i=2,5,8\ldots} U_i(\tau_j) \prod_{i=4,7,10,\ldots} U_i(\tau_j).
\end{equation}
In either case,
\begin{equation}\label{eqn:f4_activator}
U_i(\tau_j) = \sum_{\alpha, \beta, \gamma, \omega = 0}^1 P_{i-2}^\alpha P_{i-1}^\beta (H_i)^{\delta_{\alpha+\beta+\gamma+\omega,2}} P_{i+1}^\gamma P_{i+2}^{\omega}.
\end{equation}
Equation~\eqref{eqn:f4_activator} applies a Hadamard to a qubit if exactly two out of its neighbors or next-nearest neighbors are in the $\ket{1}$ state. There are no analytically known symmetries for this rule. As shown in Fig.~\ref{fig:scaling} in Sec.~\ref{subsec:complexityscaling}, our methods discover previously-unknown, symmetry-protected subspaces, indicating a hitherto hidden symmetry of the system.

\begin{figure*}
\centering
\includegraphics[width=\linewidth]{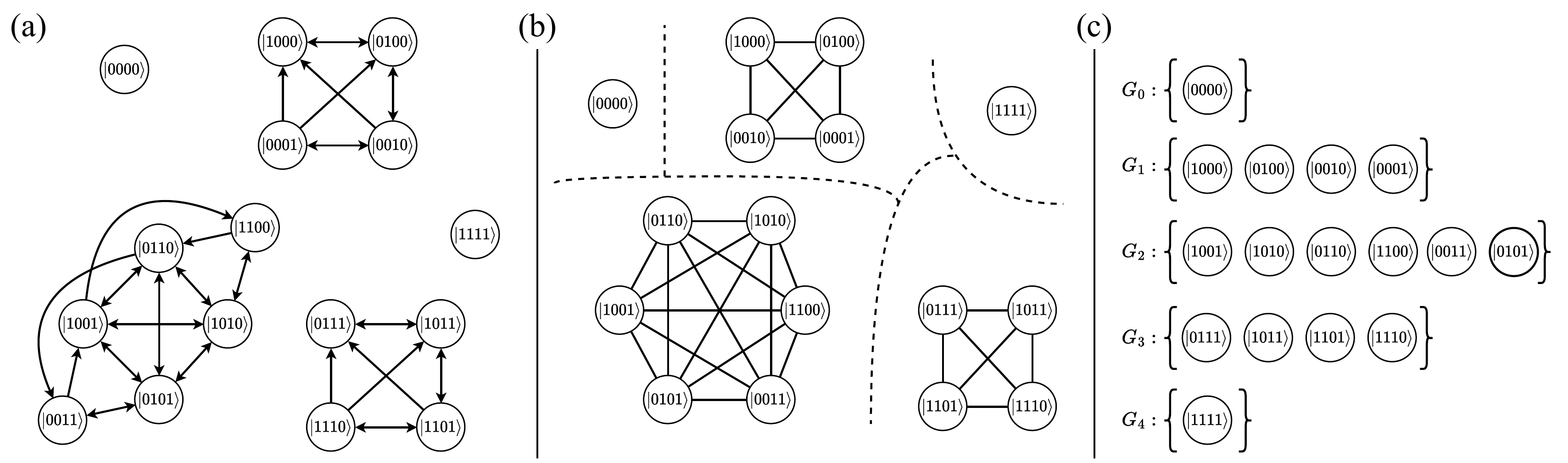}
\caption{\label{fig:workflow}Example construction of the state interaction graph $D_{U_{\text{Trot}}}$ and corresponding symmetry-protected subspaces for the one-dimensional, four-site hopping unitary in Eq.~\eqref{eq:hop_unitary}. Multi-qubit states are ordered as $\ket{q_0 q_1 q_2 q_3}$. One Trotter approximation for the unitary, $U_{\text{Trot}}$, is given in Eq.~\eqref{eq:trot_unitary}, where it should be understood that right-most operators act first. For graphical clarity we omit loops, with the understanding they are always implied. (a) State interaction graph $D_{U_{\text{Trot}}}$. Nodes represent the different 4-bit strings, and edges occur where matrix elements of $U_{\text{Trot}}$ are nonzero. Note the presence of five disconnected subgraphs. In practice, we treat each of these edges as undirected. (b) Transitive closure of $D_{U_{\text{Trot}}}$ resulting in the closed state interaction graph, which in this case corresponds exactly to the state interaction graph for the hopping unitary: $D_{U_{\text{Hop}}} = D^+_{U_{\text{Trot}}}$. Each disconnected subgraph in $D_{U_{\text{Trot}}}$ has become a disconnected complete graph in $D^+_{U_{\text{Trot}}}$. (c) Each complete graph in $D_{U_{\text{Hop}}} = D^+_{U_{\text{Trot}}}$ corresponds to a symmetry-protected subspace $G_s$. In this example, the subspace indices correspond to eigenvalues under particle number, $S=\sum_{i=0}^3 a^{\dagger}_i a_i$, conservation.
 }
\end{figure*}

\section{Graph Theory Approach to Quantum Simulations}\label{sec:graphtheory}
In this section, we show how graph theory coupled with transitive closure discovers symmetry-protected subspaces. We describe how to ``binarize'' interactions between basis states through the dynamical system by discarding amplitude and phases to solely highlight \textit{if} such an interaction exists, describing these interactions as transitive relations, and finding the states connected by a transitive relation to form the SPSs of the system. 

Directly, this method still requires unitary matrix multiplication in Hilbert space to establish basis state interactions through the dynamical system. Thus, to apply our methods, we need an efficient way to describe state-to-state interactions. We do this by creating a structure we call a \textit{string edit map} $\mathcal{L}_U$ for a unitary operator $U$, which relies on the observation that quantum systems are typically structured by local interactions. This string edit map returns the basis states available through any operations of a single local unitary operator on a single basis state in near-constant time complexity, allowing us to inexpensively find basis vectors available in local quantum dynamics.


This section will proceed as follows: in Sec.~\ref{subsec:stateinteractiongraph}, we first define the concept of a state interaction graph, whose edges indicate nonzero amplitudes on transitions between measurement basis states in the quantum simulation. Next, in Sec.~\ref{subsec:transitiveclosure} we show how transitive closure on this graph creates a cluster graph, which we call the closed state interaction graph, whose complete subgraphs are symmetry-protected subspaces. Finally, in Sec.~\ref{subsec:stringeditoperator} we show how this closed interaction graph can efficiently return the set of all measurement basis states seen by a unitary operator on a single state via the construction of the string edit map $\mathcal{L}_U$ for a unitary operator $U$.

\subsection{State interaction graph}\label{subsec:stateinteractiongraph}

We begin by defining the state interaction graph $D_U$ for a unitary operator $U$ over a Hilbert space basis $B(\mathcal{H}_d)$ and showing how to construct it. For now we leave the form of $U$ general and will specify particular forms when necessary.
\begin{definition}[State Interaction Graph]\label{def:stateinteractiongraph}
        Given a basis $B$ and a unitary operator $U$, define a vertex set $V \equiv B$ and an undirected edge set $E \equiv \{ (\ket{b}\leftrightarrow\ket{b'}) \: \: \forall \: \: \ket{b}, \ket{b'} \in B \: \: | \: \: \bra{b'}U\ket{b}\neq 0\}$. Then, the state interaction graph is defined by the ordered tuple $D_U \equiv (V, E)$. In other words, the basis states of $B$ are assigned to vertices (nodes) in the interaction graph and edges are created between vertex states only where the matrix element of the evolution unitary between the two states is nonzero.
\end{definition}
Strictly speaking, $D_U$ should be a directed graph, where an edge points from $\ket{b}$ to $\ket{b'}$ if $\bra{b'}U\ket{b}\neq 0$ and from $\ket{b'}$ to $\ket{b}$ if $\bra{b}U\ket{b'}\neq 0$. However, for a symmetry operator $S$ and time-dependent simulation unitary $U(t)$, one can show that $[S,U(t)]=0 \iff [S,U^{\dagger}(t)]=0$, meaning that symmetries of the evolution operator, and their associated protected subspaces, are invariant under time-reversal. To remain consistent with this observation, we treat every directed relationship $\bra{b'}U\ket{b}\neq 0 \Rightarrow (\ket{b}\rightarrow\ket{b'})$ as an undirected edge $(\ket{b}\leftrightarrow \ket{b'})$. This treatment is equivalent to assuming that the true, directed state interaction graph corresponding to Def.~\ref{def:stateinteractiongraph} always has a cycle that leads back to every node, such as is the case in Fig.~\ref{fig:workflow}(a). Formally, we assume that if an edge $(\ket{b}\rightarrow\ket{b'})\in D_U$, there exists a path $\{\ket{b'}\rightarrow\dots\rightarrow\ket{b}\}\subseteq D_U$. This global cyclicity assumption enables us to use the notion of transitive closure in Sec.~\ref{subsec:transitiveclosure}, and subsequently, in an uncomplicated manner that respects the time-reversal invariance of the resulting subspaces. It is justified on two points: 1.) Most simulation unitaries have a very regular, repetitive structure so that directed, acyclic state interaction graphs are likely only to arise in extremely pathological instances, and 2.) The failure of the global cyclicity assumption will only ever result in the artificial enlargement of a symmetry-protected subspace, and while such a failure leads to underconstrained subspaces, which is bad for the efficacy of e.g., post-selection, it will never result in the corruption of simulation fidelity by overconstraining or throwing out good simulation data.

To construct $D_U$, we use the following steps. First, choose a set of Hilbert space basis vectors $B(\mathcal{H}_d)=\{\ket{b}\}$. Any symmetry-protected subspaces that are discovered must be formed by the basis vectors of this basis. In the context of quantum simulation, $B$ dictates the basis in which a quantum computer will be measured. For example, parallel readout in the computational $Z$ basis will result in bit strings, $B_{\text{comp.}}=\{\ket{0\ldots00}, \ket{0\ldots01},\ldots,\ket{1\ldots11}\}$, which are Pauli $Z$-string eigenvectors. In the context of classical simulation, the basis furnishes a representation for the $d$-dimensional vector of complex amplitudes that stores the evolving many-body wavefunction. When the evolution operator is applied to a basis vector for a single time-step the resulting state $\ket{\psi(\tau)}$ has a basis vector decomposition 
\begin{equation}\label{eqn:basisvecdecomp}
    \ket{\psi(\tau)} = U(\tau) \ket{b} = \sum_{b'\in B} \ket{b'}\bra{b'}U(\tau) \ket{b} = \sum_{b'\in B}\alpha_{b'}(\tau) \ket{b'},
\end{equation}
where we suppress the time-ordering subscript $j$ in $\tau_j$ for simplicity. With this decomposition for each $\ket{b} \in B$, we create the state interaction graph $D_U$ of the operator $U(\tau)$: for each pair $\ket{b},\ket{b'}$ such that $\alpha_{b'}(\tau)\neq 0$ in Eq.~\eqref{eqn:basisvecdecomp}, we add an edge $(\ket{b}\leftrightarrow \ket{b'})$ to $D_U$. If one has a direct $d\times d$ matrix representation of $U(\tau)$ on-hand, then the adjacency matrix for $D_U$ can be read off directly as $\mathcal{A}_U =\text{bit}[U(\tau)] + \text{bit}[U^\dagger(\tau)]$, where if an entry in $U(\tau)$ or $U^\dagger(\tau)$ becomes non-zero for any value of $t$, its complex value is replaced by $1$ under the operation $\text{bit}[\ldots]$. This adjacency matrix can be formed, for example, by directly exponentiating a $d\times d$ Hamiltonian matrix with a time parameter. In practice, however, constructing or storing an entire evolution unitary in memory is costly, since the size of Hilbert space grows exponentially in the number of qubits, or spins, in the simulation: $d=2^n$. Indeed, one of the main advantages of digital quantum simulation is the ability to break global evolution unitaries into sequences of local unitaries, at the expense of introducing error, which are then implemented as quantum gates. Therefore, being able to extract symmetry-protected subspaces from consideration of local operations, rather than from the global unitary they may approximate, is of clear benefit.

Towards this end, Fig.~\ref{fig:workflow}(a) shows the state interaction graph $D_{U_{\text{Trot}}}$ for one, potentially very bad depending on $\theta$, Trotter approximation to the one-dimensional, four-qubit hopping unitary
\begin{equation}\label{eq:hop_unitary}
U_{\text{Hop}}(\theta) = e^{i \theta \sum_{i=0}^2 (X_i X_{i+1} + Y_i Y_{i+1})/2}.
\end{equation}
We take the Trotterization to be 
\begin{equation}\label{eq:trot_unitary}
U_{\text{Trot}}(\theta) = \text{iSWAP}_{01}(\theta) \times  \text{iSWAP}_{12}(\theta) \times \text{iSWAP}_{23}(\theta),
\end{equation}
where it should be understood that rightmost operators are applied first, and multi-qubit states are ordered as $\ket{q_0q_1q_2q_3}$. Notice that each $\text{iSWAP}_{i,i+1}(\theta)$ is parameterized by an arbitrary $\theta$, and as such we expect each iSWAP operation to be a fractional operation that leaves some residual state behind, i.e., the operator has an identity component. The system is small enough that the state interaction graph can be checked by hand in this case, and the main observation to be made is that it is comprised of four, disjoint subgraphs, each only containing transitions between states of fixed particle number (i.e., number of $\ket{1}$s), and that all nodes in each subgraph have a path to all other nodes in the subgraph. The zero- and four-particle states are isolated, while the one-, two-, and three-particle states form directed, incomplete, isolated subgraphs. As we will see, this ``incompleteness'' feature is a pathology of the Trotter approximation which will be rectified in Sec.~\ref{subsec:transitiveclosure} via transitive closure. It is also worth noting that constructing the state interaction graph using the Trotter-approximated unitary is not yet useful, since on a classical computer it currently still requires the storage and evolution of a $2^n$-dimensional wavefunction. We will demonstrate the utility of constructing state interaction graphs from component $k-$local operators in Sec.~\ref{subsec:stringeditoperator}.

\subsection{Defining symmetry-protected subspaces with transitive closure}\label{subsec:transitiveclosure}

Suppose we have states $\ket{b}$ and $\ket{b''}$, such that $\bra{b''}U\ket{b}=0$ and $\bra{b''}UU\ket{b}\neq 0$. This requires a transitive relation: $\bra{b''}UU\ket{b} = \sum_{b'\in B} \bra{b''}U\ket{b'}\bra{b'}U\ket{b}$, because $\ket{b}$ must first transition to an intermediate state $\ket{b'}$ to reach its final destination at $\ket{b''}$. Therefore, in our state interaction graph $D_U$, as defined in Def.~\ref{def:stateinteractiongraph}, there are edges $(\ket{b}\rightarrow \ket{b'}) ,(\ket{b'} \rightarrow \ket{b''}) \in D_U$, which will have the same transitive relation encoded in the path $\{\ket{b} \rightarrow \ket{b'} \rightarrow \ket{b''}\} \subseteq D_U $. We use this duality to make the assumption that if the edges $(\ket{b}\rightarrow \ket{b'}),(\ket{b'}\rightarrow \ket{b''}) \in D_U$, the transition amplitude $\bra{b''}UU\ket{b}\neq 0$. 

There are cases where the amplitudes of the states cancel, due to destructive interference, and break this transitive property on the level of individual basis state to basis state interactions. By ignoring the amplitudes of the basis states, we run the risk of including states in the subspace which would be removed via destructive interference. This risk comes with the benefit of efficiently knowing which states are reachable in the quantum simulation, and for our applications it does not \textit{add} any error to a simulation, as it does not break the underlying commuting subspace projection operator $P_G$; the subspaces are simply not as restrictive as they could be. See Appendix~\ref{appendix:proofdestint} for a proof. This also allows one to define the symmetry-protected subspaces to be for \textit{any} parameterization of the simulation unitary.

The transitive property exists for every state $\ket{b}$ in $D_U$, so we can take the transitive closure of the state interaction graph to create a \textit{closed} state interaction graph $D^{+}_U$. The transitive closure of an edgeset $E$ is a transitively closed edgeset $E^+$, where every pair of states $\ket{b},\ket{b'}\in E$ which can be associated by any transitive relation, in other words can be connected by a path $\{\ket{b}\rightarrow \dots \rightarrow \ket{b'}\} \subseteq E$, has an edge $(\ket{b}\rightarrow \ket{b'}) \in E^+$~\cite{lidl}.

\begin{definition}[Transitively Closed State Interaction Graph]\label{def:closedstateinteractiongraph}
        Let $D_U=(V,E)$ be a potentially non-closed state interaction graph for unitary $U$ and basis $B$. Define $V^+\equiv V$ to be the closed state interaction graph vertex (node) set and $E^+$ to be the state interaction graph edge set. An edge, $(\ket{b}\leftrightarrow\ket{b'})\in E^+$, exists in this edge set if and only if there is a path between $\ket{b}$ and $\ket{b'}$ in $D_U$, $\{\ket{b}\leftrightarrow\dots\leftrightarrow \ket{b'}\}\subseteq E$. The transitively closed state interaction graph is then defined as $D_U^+ \equiv (V^+, E^+)$.
\end{definition}

Because the original interaction graph represents single-operator state-to-state transitions, any two basis states which can discover each other through the quantum evolution have an edge in $D^{+}_U$; in other words,
\begin{equation}\label{eqn:unitaryinterpretationofgraph}
        \exists\:t\;s.t.\;\bra{b'}U(t)\ket{b}\neq 0 \Rightarrow (\ket{b}\leftrightarrow \ket{b'}) \in D^+_U,
\end{equation}
for some time/operator exponent $t$. The transitive closure of an undirected, unweighted graph is a cluster graph, or a set of complete subgraphs; as discussed at the end of Sec.~\ref{subsec:stateinteractiongraph}, we treat the graph $D_U$ as undirected just for this purpose. The transitive closure of the state interaction graph for the Trotterized unitary in Fig.~\ref{fig:workflow}(a) can be seen in Fig.~\ref{fig:workflow}(b), which in this case turns out to be state interaction graph for the original hopping unitary, or $D^+_{U_{\text{Trot}}}=D_{U_{\text{Hop}}}$. As with $D_{U_{\text{Trot}}}$, within each complete subgraph, total particle number is conserved.

If the edge $(\ket{b}\leftrightarrow \ket{b'}) \in D^{+}_U$, then $\ket{b}$ and $\ket{b'}$ share a conserved quantity of the underlying unknown symmetry. This cluster graph structure also makes it apparent that if a wavefunction is initialized as a linear combination of vectors in one subgraph $G$ of $D^{+}_U$: $\ket{\psi_0} = \sum_{b\in G} \alpha_b \ket{b}$, it will remain in that subgraph:
\begin{equation}
        \ket{\psi(t)}=U(t)\ket{\psi_0} = \sum_{b'\in G}\alpha_{b'}(t) \ket{b'}\; \forall \: t.
\end{equation}
Therefore, each complete subgraph $G$ within $D^{+}_U$ is a symmetry-protected subspace. We will formally prove that transitive closure on the state interaction graph can give an SPS with Thm.~\ref{thm:spstheorem} in Sec.~\ref{subsec:recurrence_relation}.

$D^{+}_U$ can be represented as a list of disjoint sets of nodes, where each set has implied all-to-all connectivity. This set construction can be seen for the hopping unitary in Fig.~\ref{fig:workflow}(c). Here, the index, $s$, of each subset, $G_s$ counts the number of conserved particles. Formally, we can define the associated symmetry-protected subspaces by constructing their projection operators according to Def.~\ref{def:sps}: $P_G = \sum_{b\in G} \ket{b}\bra{b}$.

\subsection{Basis state string edit map}\label{subsec:stringeditoperator}

We now have a method to identify symmetry-protected subspaces using the language of graph theory. However, actually computing these subspaces still requires the construction and manipulation of vectors and matrices in an exponentially large Hilbert space.   Recall, though, that the systems of interest are defined by Hamiltonians composed of local operations $U = \prod_i U_i $ where each $U_i$ is \textit{k-local}, meaning it only involves $k$ of the $n$ total qubits in the system, and where in general we will have $k \ll n$.  In this section we describe our mechanism for using this fact to build up SPSs efficiently with what we call the \textit{basis string edit map}.  This map enables computation of subspaces using only $k$-local operations on basis state vectors, so the computational complexity of operation with this map scales with $k$ instead of $n$.

\begin{definition}[Basis String Edit Map]\label{def:stringeditoperator}      
        Let $U_{i}(t)$  be a unitary operator that acts for a time $t$ non-trivially on $k$ of $n$ qubits, $Q_k(i) \subseteq \{q_0, \ldots, q_{n-1}\}$. The \textit{basis string edit map} $\mathcal{L}_{U_i}$ maps a basis state $\ket{b}$ to the set of basis states $\{\ket{b'}\}$ to which $\ket{b}$ can evolve after an arbitrary amount of time under $U_i(t)$. Formally, $\mathcal{L}_{U_i}(\ket{b}) = \{\ket{b'} \: | \:\exists\, t\, \bra{b'} U_i(t)\ket{b} \neq 0\}$.
\end{definition}

We can apply this construction to any unitary operator, including any $k$-local Trotter decomposition. Given a unitary operator which is a product of local operators $U = \prod_i U_i $, we form the set of local operators used in the Trotter decomposition, $\{U_i\}$. Then a basis string edit map can be formed for any subset of operators from $\{U_i\}$, as long as every operator $U_i \in U$ is included in at least one string edit map. We will use the decomposition 
\begin{equation}\label{eqn:setofeditmaps}
\mathbb{L} \equiv \{\mathcal{L}_{U_i} : U_i \in U \},
\end{equation}
which has one string edit map for each local operator in the Trotter decomposition.

In Def.~\ref{def:stringeditoperator} we have deliberately left out the exact space upon which $\mathcal{L}_{U_i}$ acts. When operating on states in a basis $B$ with $\mathcal{L}_{U_i}$, when $\text{dim}(\mathcal{L}_{U_i}) < \text{dim}(B)$, we will call $\mathcal{L}_{U_i}$ a ``substring edit map'',  and when $\text{dim}(\mathcal{L}_{U_i}) = \text{dim}(B)$, we refer to it as a ``string edit map''  or ``full string edit map".

For example, in the hopping unitary, $U_{\text{Hop}}$ given in Eq.~\eqref{eq:hop_unitary}, we can define $\mathcal{L}_{U_{\text{Hop}}}$ which would give $\mathcal{L}_{U_{\text{Hop}}}( \ket{0111}) = G_3$ where $G_3$ is in Fig.~\ref{fig:workflow}(c). More usefully, we can take the Trotterization of Eq.~\eqref{eq:hop_unitary} given by Eq.~\eqref{eq:trot_unitary}, and define string edit maps for the local iSWAP operators, $\mathbb{L}_{\text{Trot}} = \{\mathcal{L}_{\text{iSWAP}_{0,1}(\theta)}, \mathcal{L}_{\text{iSWAP}_{1,2}(\theta)}, \mathcal{L}_{\text{iSWAP}_{2,3}(\theta)}\}$. In Eq.~\eqref{eq:trot_unitary}, each local operator $U_i = \text{iSWAP}_{i,i+1}(\theta)$ would generate the corresponding edit map $\mathcal{L}_{\text{iSWAP}_{i,i+1}(\theta)}$, which operate as, e.g., $\mathcal{L}_{\text{iSWAP}_{1,2}(\theta)}(\ket{0\textbf{10}0})=\{\ket{0\textbf{01}0}, \ket{0\textbf{10}0}\}$. Here we have written the operators as acting on the full $2^n$-dimensional Hilbert space and highlighted in bold the qubits that are part of the $2^k$ dimensional subset of this space upon which $\mathcal{L}_{\text{iSWAP}_{1,2}(\theta)}$ acts. This k-local string edit map does not require information about any states besides those at the relevant indices, $1$ and $2$ in this case. 

Algorithm~\ref{alg:computeeditfunc} creates $\mathcal{L}_{U_i}$ by taking transitive closure of the adjacency matrix $\mathcal{A}_{U_i}$ of $D_{U_i}$ via boolean matrix multiplication~\cite{fischer1971}.  This algorithm requires  $O(2^{3k})$ time to compute $\mathcal{L}_{U_i}$, $O(2^k)$ space to store it, and $O(1)$ time to use $\mathcal{L}_{U_i}$. 
\begin{algorithm}[H]
\caption{Create a string edit map $\mathcal{L}_{U_i}$}\label{alg:computeeditfunc}
\begin{algorithmic}
        \Require{Local unitary $U_i$, orthonormal basis $B_i=\{\ket{b}\}$}
        \Ensure{$\text{Dim}(B_i)=\text{Dim}(U_i) = 2^k$ and $\text{Col}(U_i) = \text{Col}(B_i)$}
        \State $\mathcal{A}_{U_i} \gets \text{bit}[U_i] + \text{bit}[U_i^\dagger]$ \Comment{\parbox[t]{.5\linewidth}{Add $U_i^\dagger$ to make the adjacency matrix undirected}}
        \State $\mathcal{A}_{U_i}' \gets \text{bit}[\mathcal{A}_{U_i}^2]$
        \While{$\mathcal{A}_{U_i}' \neq \mathcal{A}_{U_i}$}
            \State $\mathcal{A}_{U_i} \gets \mathcal{A}_{U_i}'$
            \State $\mathcal{A}_{U_i}' \gets \text{bit}[\mathcal{A}_{U_i}^2]$
        \EndWhile
        
        \For{$\ket{b} \in B_i$}
            \For{$\ket{b'} \in B_i$}
                \If{$\mathcal{A}_{U_i}[\ket{b}, \ket{b'}] = 1$}
                    \State $\mathcal{L}_{U_i}(\ket{b}) \gets \mathcal{L}_{U_i}(\ket{b}) \cup \ket{b'}$
                \EndIf
            \EndFor
        \EndFor\\
        \Return{$\mathcal{L}_{U_i}$} 
\end{algorithmic}
\end{algorithm}
\noindent
Algorithm~\ref{alg:computeeditfunc} is computationally trivial to compute for small unitary operators ($k=2$ in the $\text{iSWAP}$ example), but very expensive for the large unitary operators encountered in quantum simulations. Throughout the rest of this paper, we will use $\mathcal{L}_{U_i}$ defined on small $k$ to compute our subspaces in order to keep a small overhead (thus, in our terminology, we will always be talking about ``substring edit maps"). 


The object $\mathbb{L}$ lets us create a version of the state interaction graph from Def.~\ref{def:stateinteractiongraph} which we call the \textit{string interaction graph} $D_\mathbb{L}$; this construction compactly shows the state-to-state interactions in the  algorithms presented in Sec.~\ref{sec:algorithm1} and~\ref{sec:algorithm2}. The mapping itself is defined as $\mathbb{L}(\ket{b}) \equiv \{\mathcal{L}_i(\ket{b}) \: : \: \mathcal{L}_i \in \mathbb{L}\}$.
\begin{definition}[String Interaction Graph]\label{def:basisstringinteractiongraph}
    Let $D_\mathbb{L}$ be a graph defined by an ordered tuple $D_\mathbb{L} \equiv (V_\mathbb{L},E_\mathbb{L})$ and $\mathbb{L} \equiv \{\mathcal{L}_{U_i} \: : \: U_i \in U\}$ for some Trotterized unitary $U=\prod U_i$. The vertex set is then given by $V_\mathbb{L} \equiv B(\mathcal{H}_d)$ and the edgeset is given by $E_\mathbb{L} \equiv \{(\ket{b}\leftrightarrow\ket{b'}) \: : \: \ket{b},\ket{b'} \in B \: \mathrm{and} \: \ket{b'} \in \mathbb{L}(\ket{b})\}$.
 \end{definition}
Notice that, while edges in $D_U$ can capture the action of multiple operators at once, each edge in $D_{\mathbb{L}}$ is the action of only a single operator; therefore, if $D_U\equiv(V,E)$ is the state interaction graph and $D_\mathbb{L}\equiv(V,E_\mathbb{L})$ is the string interaction graph, then $E_\mathbb{L} \subseteq E$. Later, we will show with the proof in Appendix~\ref{appendix:prooftc} for Thm.~\ref{thm:spstheorem} that despite this inequality, the transitive closure of the graph $D_\mathbb{L}^+$, is equivalent to the transitively closed state interaction graph $D_\mathbb{L}^+ \equiv D_U^+$.

\begin{figure*}
\centering
\includegraphics[width=\linewidth]{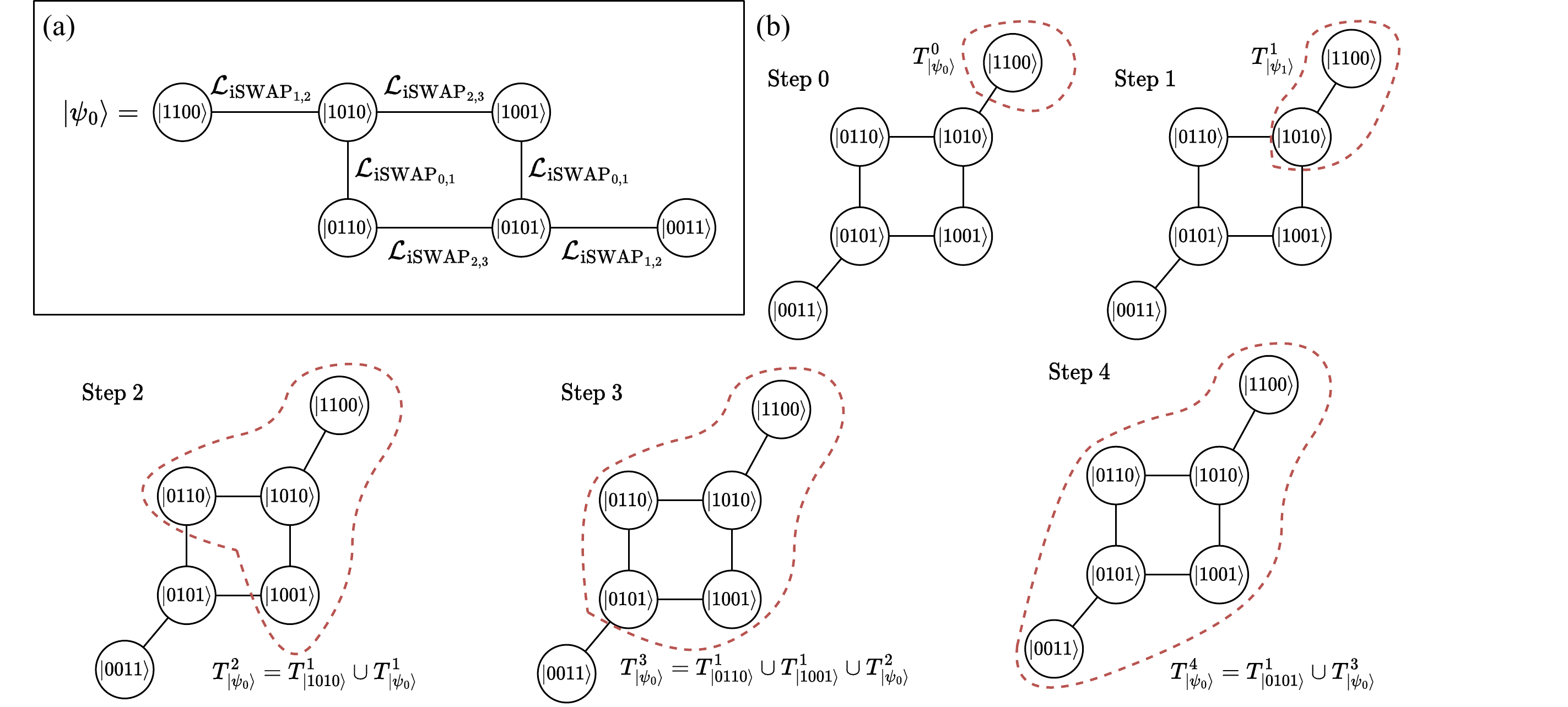}
\caption{\label{fig:recursiverelation} This figure shows how we can incrementally build the symmetry-protected subspace for $U=\prod_{i=0}^2\text{iSWAP}_{i,i+1}(\theta)$. (a) A single un-closed symmetry-protected subspace, or a subgraph of $D_\mathbb{L}$ for an iSWAP network. All operations $\mathbb{L}_{\text{Trot}}=\{\mathcal{L}_{\text{iSWAP}_{0,1}(\theta)}, \mathcal{L}_{\text{iSWAP}_{1,2}(\theta)}, \mathcal{L}_{\text{iSWAP}_{2,3}(\theta)}\}$ which are non-identity on each node are shown as an edge. Notice the similarities between this graph and the $s=2$ subgraph of Fig.~\ref{fig:workflow}(a): both have the same vertex set, but the edgeset of $D_\mathbb{L}$ is a subset of the edgeset in $D_U$. (b) Follow the recursion relation in Eq.~\eqref{eqn:recursionrelation} to iteratively build the symmetry-protected subspace $G_{\ket{1100}}$, starting from $\ket{\psi_0}=\ket{1100}$. Even though the graph in (a) is not equivalent to the graph in Fig.~\ref{fig:workflow}(a), their transitively closed graphs are equivalent, as can be seen by the vertices covered by the red line.
 }
\end{figure*}

\section{Algorithm: creating the symmetry-protected Subspace of an Initial State}\label{sec:algorithm1}
We now turn to the construction of the entire symmetry-protected subspace of a given initial state using $k$-local substring edit maps and transitive closure, a construction consistent with the observation that if a symmetry exists locally everywhere in a quantum circuit, then it will also manifest globally \cite{marvian2020locality}. The algorithm by which we do so works by establishing a recurrence relation for computing new states, associated by symmetry-protection to the initial state, in the simulation. This recurrence relation furnishes an efficient way to compute the transitive closure of select subgraphs of the entire Hilbert space, with no extraneous information. For a unitary simulation operator decomposed into enumerated local operations, $U(\tau_j) = \prod_{i=1}^{m} U_i(\tau_j)$, where each $U_i(\tau_j)$ is $k$-local, we will show that a symmetry-protected subspace $G_{\ket{\psi_0}}$ of the initial state $\ket{\psi_0}$ costs $O((m+1)\times|G_{\ket{\psi_0}}|)$ to compute and $O(|G_{\ket{\psi_0}}|)$ to store with a breadth-first search~\cite{Skiena2008}. We will give the algorithm for the case where the initial state $\ket{\psi_0}$ is a single measurement basis state (i.e., a product state). If $\ket{\psi_0}$ is a linear combination of measurement basis vectors, the algorithm can be repeated for each basis vector in the sum; this does not impact the asymptotic performance of the algorithm, as it only makes the computed subspace bigger.

We enumerate every state in a symmetry-protected subspace by transitively closing subgraphs created by the local basis substring edit maps established in Sec.~\ref{subsec:stringeditoperator}, which return a set of basis states evolved to by their corresponding unitary operators in $O(1)$ when operating on a single basis state. We recursively build the subspace by checking the set of substring edit maps $\mathbb{L}\equiv \{\mathcal{L}_{U_i}\: : \: U_i \in U\}$ on each new state, until none are added.  This process can be seen as the transitive closure of a subgraph of the graph $D_\mathbb{L}$. 
For $U_{\text{Trot}}$ in Eq.~\eqref{eq:trot_unitary}, Fig.~\ref{fig:recursiverelation}(a) shows an example of the subgraph of $D_{\mathbb{L}}$ corresponding to the action of each $\mathcal{L}_{\text{iSWAP}_{i,i+1}(\theta)} \in \mathbb{L}_{\text{Trot}}$ starting from the initial state $\ket{\psi_0}=\ket{1100}$ (self-edges are ignored as elsewhere in the manuscript).

\subsection{Recurrence relation}\label{subsec:recurrence_relation}
To find the symmetry-protected subspace of an initial state, we begin by computing the string edit map for each $k$-local unitary, $\mathbb{L}\equiv\{\mathcal{L}_{U_i} \: : \: U_i \in U\}$. Next, we check the set of measurement basis strings generated by operating with each substring edit map on the initial state, notated $\mathbb{L}(\ket{\psi_0}) \equiv \{\mathcal{L}_i(\ket{\psi_0}) \: : \: \mathcal{L}_i \in \mathbb{L}\}$. We define the set $T_{\ket{\psi_0}}^1\equiv \{\ket{\psi_0}\}\: \cup \: \mathbb{L}(\ket{\psi_0})$. Then, for each \textit{new} state $\ket{\phi}\in \mathbb{L}(\ket{\psi_0})$, we check operations under the substring edit maps: $T_{\ket{\psi_0}}^2=\mathbb{L}(T_{\ket{\psi_0}}^1)\equiv \{\mathbb{L}(\ket{\phi}) \: : \: \ket{\phi} \in \mathbb{L}(\ket{\psi_0})\} $. This process repeats until no new states are found through the following recurrence relation.
\begin{equation}\label{eqn:recursionrelation}
        \begin{cases}
        T_{\ket{\psi_0}}^1      \leftarrow \{\ket{\psi_0}\}   \cup \mathbb{L}(\ket{\psi_0})        & \text{base case}\\
        T_{\ket{\psi_0}}^{i+1}  \leftarrow T_{\ket{\psi_0}}^{i} \cup \mathbb{L}(T_{\ket{\psi_0}}^i) & \text{recursive case}\\
        T_{\ket{\psi_0}}^{i+1} = T_{\ket{\psi_0}}^i & \text{stop condition} 
        \end{cases}
\end{equation}
The stop condition activates if no new states are found, that is, when additional operations drawn from $\mathbb{L}$ do not unveil any new states. Steps 0-4 in Fig.~\ref{fig:recursiverelation}(b) show how the recurrence relation manifests for the input state $\ket{\psi_0}=\ket{1100}$ and the set of substring edit maps generated from the Trotterization in Eq.~\eqref{eq:trot_unitary}. We then define the symmetry-protected subspace $G_{\ket{\psi_0}}$ to which the state $\ket{\psi_0}$ belongs via
\begin{equation}\label{eq:TG_equiv}
G_{\ket{\psi_0}} \equiv T_{\ket{\psi_0}}^{i+1},
\end{equation}
where the definition ``$\equiv$'' in Eq.~\eqref{eq:TG_equiv} should be taken to mean ``all basis states in $T_{\ket{\psi_0}}^{i+1}$ viewed as nodes in a complete graph''. We can create the Kleene Closure of the set of substring edit maps, denoted $\mathbb{L}^{\star}$, which is the set of all finite concatenations of substring edit maps, including the identity. Any arbitrary string of substring edit maps $\mathbb{L}^{\star}$ applied to $\ket{\psi_0}$ will result in a state in $G_{\ket{\psi_0}}$ by its definition. Hence, one can write
\begin{equation}
\mathbb{L}^{\star}(\ket{\psi_0}) = G_{\ket{\psi_0}}.
\end{equation}
We now state our main result.
\begin{theorem}\label{thm:spstheorem}
        Let $U(t)=\mathcal{O}_{op} \big[\prod_{j=1}^p \prod_{i=1}^m U_i(\tau_j)\big]$ be a quantum simulation unitary of duration $t$ as in Eq.~\eqref{eqn:timeevol}, divided into $p$ time-steps, where each $U_i(\tau_j)$ is $k$-local and is time-step parameterized by $\tau_j$, and some operator ordering (such as time-ordering) is specified. Let $\mathcal{L}_{U_i}$ be the string edit map corresponding to any available parameterization of $U_i(\tau_j)$ and the set of such maps $\mathbb{L} \equiv \{\mathcal{L}_{U_i} \: : \: i \in \{1,\ldots,m\}\}$. Let $B(\mathcal{H}_{2^n})\equiv \{\ket{b}\}$ be the basis in which computations (measurements) are being performed classically (quantumly).
        Then, given an input state $\ket{\psi_0}$, expressed in the basis $B$, if $\ket{b}\notin G_{\ket{\psi_0}}$, where $G_{\ket{\psi_0}}$ is constructed according to Eqs.~\eqref{eqn:recursionrelation}-\eqref{eq:TG_equiv}, then $\bra{b}U(t)\ket{\psi_0} = 0$.
\end{theorem}
\noindent
For a proof, see Appendix~\ref{appendix:prooftc}. Note that for $\bra{b_f}U(t)\ket{\psi_0}$ to vanish under these conditions, $G_{\ket{\psi_0}}$ must satisfy Def.~\ref{def:sps} as demonstrated in Lem.~\ref{lem:thm1_connection}. In other words, for $G_{\ket{\psi_0}}$ to be able to exclude particular basis states for arbitrary evolution times, it must be a symmetry-protected subspace.
Note that Thm.~\ref{thm:spstheorem} immediately provides two corollaries. 

\bigskip
\noindent
\textit{Corollary 1 (Post-Selection).} For simulation on an idealized, noise-free quantum computer, if $\ket{b_f}\notin G_{\ket{\psi_0}}$ then $||\bra{b_f}U(t)\ket{\psi_0}||^2 = 0$. Hence, if the state $\ket{b_f}$ is measured in the output of a noisy quantum device, it can be assumed that the state arose as a result of error, and may be discarded.

\bigskip
\noindent
\textit{Corollary 2 (Global Subspace).} We assumed a Trotterized form for $U(t)$ in the statement and proof of Theorem~\ref{thm:spstheorem}. However, we can formally recover the corresponding global simulation unitary by taking the limit $p\rightarrow \infty$ where $\tau_j = jt/p$ in the time-ordered case and $\tau_j = t/p \: \forall j$ when time-ordering is unnecessary (such as when the Hamiltonian is time-independent). Nothing in the proof of Theorem~\ref{thm:spstheorem} relies upon the finiteness of $p$ or discreteness of the corresponding time differential $t/p$. Therefore, our result holds for global simulation unitaries as well. This implies that one can reduce the resource requirements in classical simulations of $U(t)$ by only evolving basis states $\ket{b_f}\in G_{\ket{\psi_0}}$.

\subsection{Pseudocode}
With an understanding of the recurrence relation in Eq.~\eqref{eqn:recursionrelation} and how it can compute symmetry-protected subspaces, we present an algorithm which can enumerate these subspaces using a breadth-first search. Breadth-first search to enumerate an entire graph $(V,E)$ has computational complexity $O(|V|+|E|)$. In our implementation, there are $|G_{\ket{\psi_0}}|$ vertices and we check for $m$ edges at each vertex, giving $O(m\times |G_{\ket{\psi_0}}|)$ edges in the entire graph. Thus, our breadth-first search to enumerate the symmetry-protected subspace is $O(|G_{\ket{\psi_0}}| + m\times|G_{\ket{\psi_0}}|) = O((m+1)\times|G_{\ket{\psi_0}}|)$
\begin{algorithm}[H]
\caption{Enumerate symmetry-protected subspace $G_{\ket{\psi_0}}$ with a breadth-first search}\label{alg:makesps}
\begin{algorithmic}
        \Require{String edit maps $\mathbb{L}$ of the simulation operator, initial state $\ket{\psi_0}$}
        \State $G \gets \{\ket{\psi_0}\}$
        \State Let $Q$ be a first-in-first-out queue
        \State $Q\text{.enqueue(}\psi_0)$
        \While{$Q \neq \emptyset$}
            \State $\ket{b} \gets Q\text{.dequeue()}$
            \For {$\ket{b'} \in \{\mathcal{L}_{U_i}(\ket{b}) \: : \: \mathcal{L}_{U_i} \in \mathbb{L}\}$}
                \If{$\ket{b'} \notin G$}
                    \State $G \gets G \cup \ket{b'}$
                    \State $Q\text{.enqueue(}\ket{b'})$
                \EndIf
            \EndFor
        \EndWhile\\
        \Return{$G_{\ket{\psi_0}} \gets G$}
\end{algorithmic}
\end{algorithm}
Algorithm~\ref{alg:makesps} uses the set $G$, which is eventually the symmetry-protected subspace, to track which states have already been added to $Q$ during the runtime of the algorithm and prevent them from being checked more than once. If this set uses the hash of the basis state's bitstrings, insertion and search will be average-case $O(1)$. Getting the set of single-operator transitions $T_{\ket{b}}^1$ is $O(m)$, where $m$ is the number of unitary operators in the system, for a single state $\ket{b}$; while not all $m$ substring edit maps will provide an edge, as many might act as identity on the state $\ket{b}$, each string edit map must still be checked. This set is computed for each state discovered, and each state discovered is never added to the queue $Q$ more than once, which confirms our original complexity analysis of $O((m+1)\times |G_{\ket{\psi_0}}|)$.

This algorithm computes a set equivalent to that described by Eq.~\eqref{eqn:recursionrelation}. See Appendix~\ref{appendix:prooftc} for a proof that this set is a symmetry-protected subspace.

\begin{figure}[h]
        \centering
        \includegraphics[width=\linewidth]{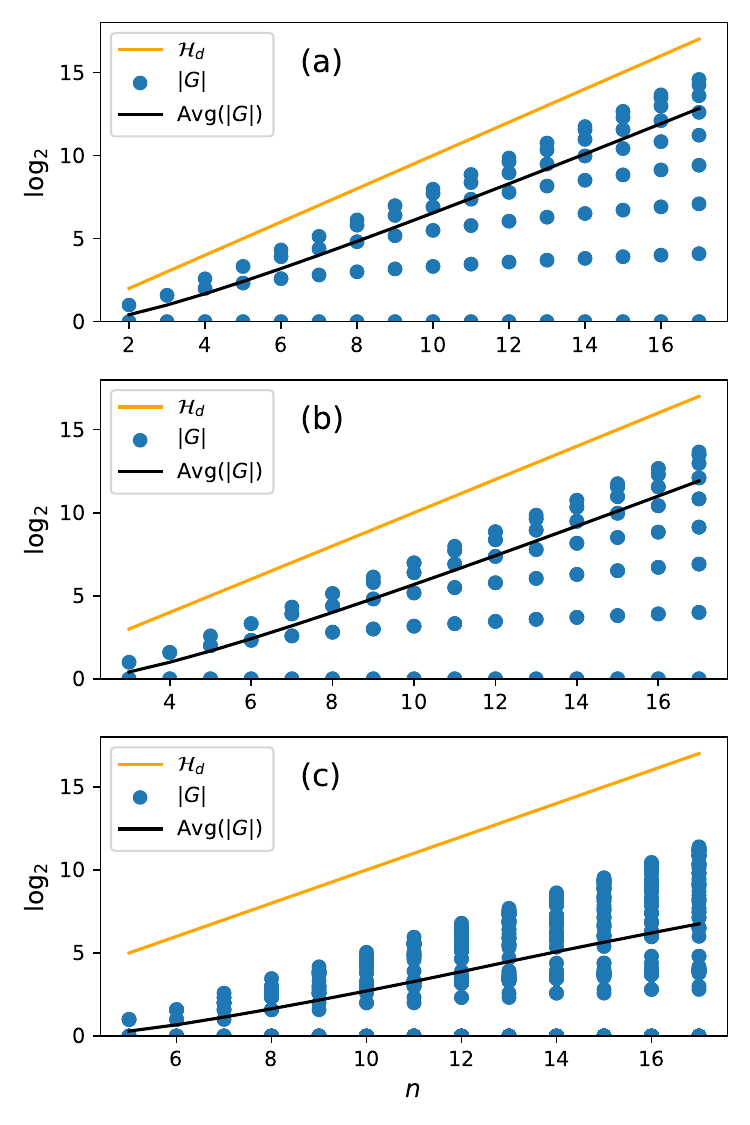}
        \caption{The number of states in symmetry-protected subspaces for each model, on a $\log_2$ scale. Each blue dot is the size of an individual subspace. (a) Heisenberg-XXX. (b) $T_6$ QCA. (c) $F_4$ QCA, which has no known symmetry operator.}
        \label{fig:scaling}
\end{figure}

\subsection{Usage and limits}\label{subsec:complexityscaling}
As stated in Sec.~\ref{sec:introduction}, post-selection for quantum simulations is the aim of our methods. In order to perform post-selection with the algorithm from this section, the entire symmetry-protected subspace $G_{\ket{\psi_0}}$ must be known, and each measurement $\ket{b_f}$ is verified via
\begin{equation}
        \begin{cases}
                \ket{b_f}\in G_{\ket{\psi_0}} & \text{assumed no error}\\
                \ket{b_f}\notin G_{\ket{\psi_0}} & \text{known error.}
        \end{cases}
\end{equation}
Because this requires the computation of the entire subspace, it can still become computationally intractable. As mentioned in Sec.~\ref{sec:introduction}, we benchmark this algorithm against the Heisenberg-XXX, T6 QCA, and F4 QCA quantum simulations. The computationally limiting factor of this algorithm is the size of the symmetry-protected subspace, as it has time complexity $O((m+1)\times|G_{\ket{\psi_0}}|)$ and spatial complexity $O(|G_{\ket{\psi_0}}|)$. See Fig.~\ref{fig:scaling} for a depiction of the size of the subspaces, up to 17 qubits. Figure~\ref{fig:scaling}(c) is especially significant because, as alluded to in Sec.~\ref{subsubsec:f4qca}, this model previously had no known conservation laws, but this data shows the partitioning of Hilbert space into symmetry-protected subspaces. By examining Fig.~\ref{fig:scaling}, we can see that the worst case of each subspace size, $\max(\{|G|\})$, is $\log_2(\max(\{|G|\}))\approx \log_2(|\mathcal{H}_d|) - k$ when the model is comprised of $k$-local operations. The important thing to see is that the size of symmetry-protected subspaces still scales exponentially in the worst and average cases; despite them being smaller than the full Hilbert space, they are only linearly smaller. We address this stop to post-selection with another algorithm in Sec.~\ref{sec:algorithm2}.

\section{Algorithm: Verification of a Shared symmetry-protected Subspace}\label{sec:algorithm2}

Because the worst-case size of a symmetry-protected subspace is still exponential in the number of qubits, Alg.~\ref{alg:makesps} is only applicable at relevant system sizes for simulations where $G_{\ket{\psi_0}}$ is not exponentially large. Thus, to generalize the usability of symmetry-protected subspaces to any simulation, we present an alternative algorithm in this section that uses an efficient, but greedy, heuristic. Instead of computing the entire exponentially large SPS, this method performs an efficient search for a path of substring edit maps to connect the initial state $\ket{\psi_0}$ and a measured state $\ket{b_f}$ in the graph $D_\mathbb{L}$. The heuristic nature of this algorithm means that it produces only approximate results (albeit with a degree of approximation that can be continuously improved at the cost of more time complexity), and since its runtime scales favorably, it can be used to check the output of quantum simulations well beyond the scale where other classical methods become intractable.


Naive search algorithms in a graph traditionally use a breadth-first search from an initial vertex~\cite{5219222}, which is what we described in Alg.~\ref{alg:makesps} to enumerate the symmetry-protected subspace through transitive closure; as stated, this is too computationally expensive in many cases.

Our greedy algorithm works as follows: Following all edges (application of the set of substring edit maps $\mathbb{L}$) from a state generates a set of possible states accessible from that state. Working both forward from $\ket{\psi_0}$ and backward from $\ket{b_f}$, we have found a \textit{path} when steps from both directions lead to a shared element. To detect whether this has happened, we use a simple observation about ordered sets: two sets are the \textit{same} set if they have the same minimal (or maximal) element. A natural order for sets of quantum computational basis states (i.e., binary strings) is just the integer they encode. So, our greedy algorithm works by building two sets (one starting from the initial state and one from the measured state) and comparing them; these sets are subsets of the symmetry-protected subspace corresponding to the initial and final states. If these sets have any common elements (which we check in constant time by looking at their \textit{minimal} elements), then we have found a path connecting the initial and measured state, and they occupy the same SPS.

In our methods, post-selection's accept/reject decision is made by checking for a measurement result in the symmetry-protected subspace of the initial state. Using the rationale outlined above, we can declare with certainty that two states share a SPS when their paths collide. On the other hand, if the two paths terminate in dissimilar minima we \textit{assume} the initial and measured states inhabit disjoint subspaces. We say ``assume'' because the accuracy of our heuristic differs depending on the simulation; when the locally-minimal choice at each step from either state does not build a path to the true minimum element, the two states may build distinct paths while occupying the same SPS.

\subsection{Searching in the string interaction graph} \label{subsec:pathfinding}
This algorithm builds a single branch, following local minima in depth-limited breadth-first searches, from a single starting vertex in the $D_\mathbb{L}$ graph, where each vertex $\ket{b}$'s edgeset is given by $\mathbb{L}(\ket{b})$.

Throughout this section we use the integers encoded by the bitstrings $b$ of basis states $\ket{b} \in B(\mathcal{H}_{2^n})$, which provides a natural ordering to states. Let the function $\min(A)$ on a set of states $A$ return the state with the smallest encoded integer in that set: 
\begin{equation}
\min(A) = \ket{b} \: s.t. \: b \leq b',\, \forall \, \ket{b'} \in A .
\end{equation}
To find the minimal element of the symmetry protected subspace, we start with an element $\ket{b_0} =\ket{\psi_0}$ or $\ket{b_f}$ and build a set towards the minimum of $G_{\ket{b_0}}$, notated $\min(G_{\ket{b_0}})$, with locally optimal decisions. At each step $\ket{b_{\text{curr}}}$ in the search, we compute $T_{\ket{b_{\text{curr}}}}^\mu$ from Eq.~\eqref{eqn:recursionrelation}, which is a breadth-first search to depth $\mu$, or every state in $\mathbb{L}^\mu(\ket{b_{\text{curr}}})$. Then, the starting point for the next step, $\ket{b_{\text{next}}}$, is the state with the smallest binary encoded integer in $T_{\ket{b_{\text{curr}}}}^\mu$, notated $\min(T_{\ket{b_{\text{curr}}}}^\mu)$. Repeat until the set $T_{\ket{b_{\text{curr}}}}^\mu$ does not offer a state smaller than $\ket{b_{\text{curr}}}$. Let this process be represented by the recursive function $\chi$:
\begin{equation}\label{eqn:iterativepathbuild}
    \chi(\ket{b_0}, \mu)=
        \begin{cases}
        \ket{b_{\text{curr}}} &\text{if} \ket{b_{\text{curr}}} = \min(T_{\ket{b_{\text{curr}}}}^\mu)\\
        \chi(\min(T_{\ket{b_{\text{curr}}}}^\mu), \mu)  & \text{otherwise}
        \end{cases}.
\end{equation}
To reiterate, because each step is given by applications of the substring edit maps, $\chi(\ket{b_0},\mu)$ and $\ket{b_0}$ must share a symmetry-protected subspace, i.e., $\chi(\ket{b_0},\mu) \in G_{\ket{b_0}}$. Therefore, when $\chi(\ket{\psi_0},\mu) = \chi(\ket{b_f},\mu)$ there \textit{must} be a sequence of substring edit maps between $\ket{\psi_0}$ and $\ket{b_f}$; i.e., a there is a path $\{\ket{\psi_0} \leftrightarrow \dots \leftrightarrow \ket{b_f}\} \subseteq D_\mathbb{L}$ and $\ket{b_f} \in \mathbb{L}^*(\ket{\psi_0})$. 


When both searches conclude in the true minimal element of their corresponding symmetry-protected subspaces, $\chi(\ket{\psi_0},\mu) = \min(G_{\ket{\psi_0}})$ and $\chi(\ket{b_f},\mu)=\min(G_{\ket{b_f}})$, we can conclude with certainty that the states do or do not inhabit the same subspace if $\chi(\ket{\psi_0},\mu) = \chi(\ket{b_f},\mu)$ or $\chi(\ket{\psi_0},\mu) \neq \chi(\ket{b_f},\mu)$. However, if either search does \textit{not} conclude in their targeted minimal state, our conclusions can be wrong. Suppose $\ket{b_f} \in G_{\ket{\psi_0}}$, which means that $\min(G_{\ket{\psi_0}})= \min(G_{\ket{b_f}})$, but the search result $\chi(\ket{b_f},\mu)$ finds a false minima of $G_{\ket{b_f}}$, meaning $\chi(\ket{b_f},\mu)\neq \min(G_{\ket{b_f}})$. When post-selecting with our procedure under these conditions, state $\ket{b_f}$ would be wrongly rejected because the underlying assumption for this heuristic, that $\chi(\ket{b_0},\mu) = \min(G_{\ket{b_0}})$, is wrong. In practice, the true minimal element of the symmetry-protected subspace is unknown; thus, the assertion that $\chi$ finds the minimum is always an assumption.

Consequently, if the result of the two searches is a collision, $\chi(\ket{\psi_0},\mu) = \chi(\ket{b_f},\mu)$, we \textit{know} the two states must share a symmetry-protected subspace, even if the searches are at a false minima. If the two searches do not find a common element, we \textit{assume} the initial and measured states occupy separate symmetry-protected subspaces. As such, measurement outcomes that lie within the intial state's symmetry-protected subspace can be rejected. Thus, each measurement $\ket{b_f}$ is verified with
\begin{equation}\label{eq:ps_with_path}
        \begin{cases}
                \chi(\ket{b_f},\mu) = \chi(\ket{\psi_0},\mu) & \text{assumed no error}\\
                \chi(\ket{b_f},\mu) \neq \chi(\ket{\psi_0},\mu) & \text{assumed error}
        \end{cases}.
\end{equation}
Formally, we can state that the protected subspace formed by $\chi$ around an initial state $\ket{\psi_0}$ is 
\begin{equation}
    P_G' = \sum_{b} \ket{b}\bra{b} \; s.t. \; \chi(\ket{b}, \mu) = \chi(\ket{\psi_0}, \mu)
\end{equation}
which is approximately equal to the true symmetry-protected subspace, $P_G \approx P_G'$. 

The confidence of this assumption depends on the quantum system being studied, and the depth $\mu$ of the local breadth-first searches. Thus, we present arguments to support it for our three exemplar systems: Heisenberg-XXX, $T_6$ QCA, and $F_4$ QCA as outlined in Sec.~\ref{subsec:quantumsimulation}. For the Heisenberg-XXX model this assumption is always correct at $\mu=1$: systems whose substring edit maps $\mathbb{L}$ are isomorphic to nearest-neighbor SWAP substring edit maps $\mathbb{L}_{\text{SWAP}}$, such as the Heisenberg-XXX model, provably find $\min(G_{\ket{b_0}})$ using $\chi(\ket{b_0}, \mu=1)$, as shown in Appendix~\ref{appendix:proofofswap}. We also numerically find that $\chi(\ket{b_0},\mu=2)$ is exact for the $T_6$ QCA model. Out of our three studied systems the $F_4$ QCA demonstrates the worst accuracy; $\chi(\ket{b_0},\mu=9)$ still returns false minima, which potentially causes a false rejection. In Sec.~\ref{subsubsec:pathaccuracydata} we analyze these properties of the $T_6$ and $F_4$ QCAs more closely. In Sec.~\ref{subsec:probps} we show that post-selection is still effective when $\chi$ causes a false rejection to occur. 

The computational performance of this method is also model-dependent. Let $|Depth|$ be the depth of our search, or the number of intermediate states traversed to (the second line in Eq.~\eqref{eqn:iterativepathbuild}) before finding a locally optimal vertex. This means the sequence of substring edit maps between $\ket{b_0}$ and $\chi(\ket{b_0},\mu)$ is length $O(\mu\times|Depth|)$, and the computational complexity of Eq.~\eqref{eqn:iterativepathbuild} is $O(m^\mu \times |Depth|)$.

Because the reliability and computational complexity of this post-selection method varies, before applying the algorithm to a given quantum simulation, benchmarks should be performed on a small version of the system to understand the $\mu$ required to obtain reliable results.

\subsection{Pseudocode}\label{subsec:pathfinding_pseudocode}

Here we describe the algorithm, Alg.~\ref{alg:greedypath}, to find the minimum binary encoded state in the symmetry-protected subspace of a given basis state. This algorithm starts at a state $\ket{b_0}$ and follows the local minima of limited breadth-first searches in an attempt to find the set's minimal element, resulting in $\chi(\ket{b_0},\mu)$.

Notice that the innermost \textbf{while} loop is borrowed from Alg.~\ref{alg:makesps} to compute each $T_{\ket{b_{\text{curr}}}}^\mu$, where $\ket{b_{\text{curr}}}$ is the current best minimum state. This time, the queue $Q$ uses the ordered tuple $(\ket{b'},\eta)$ as its elements, where $\ket{b'}$ is the current, un-checked, state in the breadth-first search and $\eta$ is the number of edges (i.e., depth) that $\ket{b'}$ is away from the current optimal solution $\ket{b_{\text{curr}}}$. Each node checked in this loop has $O(m)$ new edges, and we check a maximum depth of $\mu$, giving the inner loop computational complexity $O(m^\mu)$. The set $T^\mu_{\ket{b_\text{curr}}}$ is checked for every node traversed to in the search, and our search traverses to $|Depth|$ elements, which is how we arrive at an overall complexity of $O(m^\mu \times|Depth|)$ to compute $\chi(\ket{b_0},\mu)$.
\begin{algorithm}[H]
\caption{Greedy pathfinding to the minimum binary encoded state}\label{alg:greedypath}
\begin{algorithmic}
        \Require{Path starting state $\ket{b_0}$, substring edit maps $\mathbb{L}$, search limiting integer $\mu$}
        \State{$\ket{b_{\text{next}}} \gets \ket{b_0},\,\ket{b_{\text{curr}}} \gets$  null }
        \While{$\ket{b_{\text{curr}}} \neq \ket{b_{\text{next}}} $} 
            \State $T_{\ket{b_{\text{curr}}}} \gets \{\ket{b_{\text{curr}}}, \ket{b_{\text{next}}}\}$
            \State $\ket{b_{\text{curr}}} \gets \ket{b_{\text{next}}}$
            \State Let $Q$ be a queue
            \State $Q\text{.enqueue}(\ket{b_{\text{curr}}}, 0)$ \Comment{\parbox[t]{.5\linewidth}{Each element in $Q$ is tuple representing (state, depth)}}
            \While{$Q \neq \emptyset$} \Comment{\parbox[t]{.5\linewidth}{Breadth-first search centered around $\ket{b_{\text{curr}}}$}}
                \State $\ket{b'}, \eta \gets Q\text{.dequeue()}$
                \If{$\eta = \mu$} 
                    \textbf{break} \Comment{Reached the depth limit}
                \Else 
                    \For{$\ket{b''} \in \{\mathcal{L}_{U_i}(\ket{b'}) \: : \: \mathcal{L}_{U_i} \in \mathbb{L}\}$}
                        \If{$\ket{b''} \notin T_{\ket{b_{\text{curr}}}}$}
                            \State $Q\text{.enqueue}(\ket{b''}, \eta + 1)$
                            \State $T_{\ket{b_{\text{curr}}}} \gets T_{\ket{b_{\text{curr}}}} \cup \ket{b''}$
                        \EndIf
                    \EndFor
                \EndIf
            \EndWhile
            \State $\ket{b_{\text{next}}} \gets \min(T_{\ket{b_{\text{curr}}}})$
        \EndWhile\\
        \Return $\chi(\ket{b_0}, \mu) = \ket{b_{\text{curr}}}$
\end{algorithmic}
\end{algorithm}
To post-select with Alg.~\ref{alg:greedypath}, run it once for $\ket{\psi_0}$ to get $\chi(\ket{\psi_0},\mu)$, again for each measurement $\ket{b_f}$ to get $\chi(\ket{b_f},\mu)$, and check the results in Eq.~\eqref{eq:ps_with_path}. 

As an aside, if all states computed in the search, $T^\mu_{\ket{b_{\text{curr}}}}$, are cached with the result of their search, subsequent executions of the algorithm can be preempted with their previously calculated result if a state is found inside the cache.

\subsection{Benchmarking data for Algorithm~\ref{alg:greedypath}} \label{subsec:greedypathdata}
This subsection presents benchmark data for Alg.~\ref{alg:greedypath} that demonstrates its practical applicability for our example systems. These benchmarks include heuristic reliability and search depth. We supplement this data with proofs or explicit analytical forms when possible. When this is not possible, we instead rely on extrapolation from the data to inform the asymptotics of these quantities. We find that for the Heisenberg-XXX model, the heuristic is provably exact (Appendix~\ref{appendix:proofofswap}) and the search depth has an explicit equation (Appendix~\ref{appendix:pathlengthcalcalc}), while the $T_6$ and $F_4$ QCA models rely on data-driven intuition for these quantities.

\begin{figure}[htbp]
        \centering
        \includegraphics[width=\linewidth]{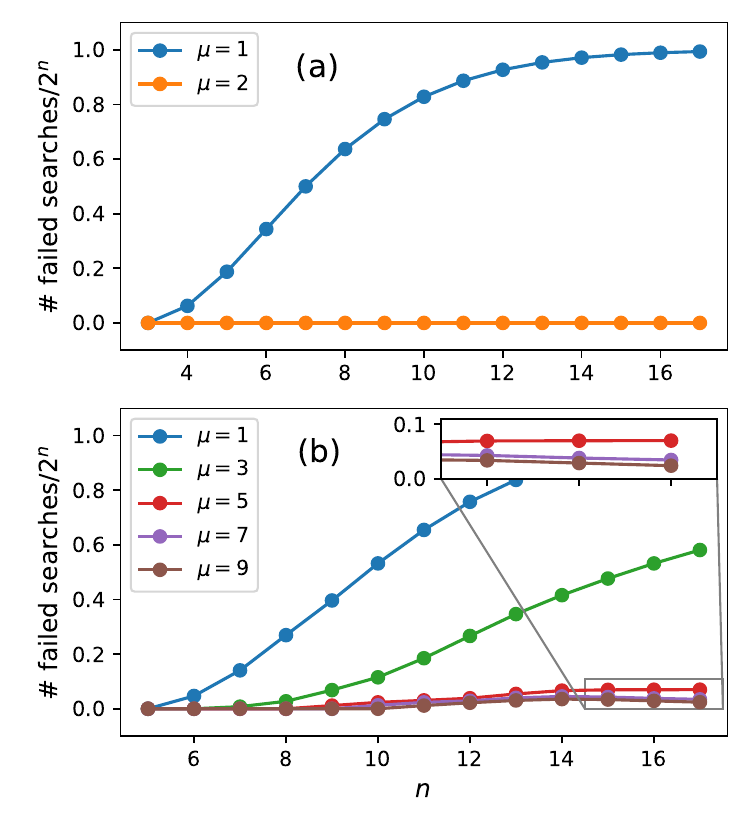}     
        \caption{The proportion of states $\ket{b}$ where $\chi(\ket{b_0},\mu) \neq \min(G_{\ket{b}})$ via Alg.~\ref{alg:greedypath}, called ``failed searches'', at different system sizes, for each $\mu$ used. (a) $T_6$ QCA. At $\mu=2$, Alg.~\ref{alg:greedypath} becomes exact for this system. (b) $F_4$ QCA. High $\mu$ can still fail, but does asymptotically better for $\mu\geq5$.}
        \label{fig:falseminima}
\end{figure}

\subsubsection{Search reliability}\label{subsubsec:pathaccuracydata} 
If searching for the minimal element from $\ket{b_0} =\ket{\psi_0}$ or $\ket{b_f}$ fails, because $\chi(\ket{b_0},\mu) \neq \min(G_{\ket{b_0}})$, then it is possible to falsely assume that two states occupy disjoint subspaces. As such, we examine the reliability of the searching function $\chi$ on each exemplar system.

In the Heisenberg-XXX model, the base case for $\chi$, as shown in Eq.~\eqref{eqn:iterativepathbuild}, only activates at a state which is always the minimum binary-encoded state in the symmetry-protected subspace. This is because the set of substring edit maps for the Heisenberg-XXX system is isomorphic to the set of substring edit maps for the nearest-neighbor SWAP network: $\mathbb{L}_{\text{HeisXXX}} = \mathbb{L}_{\text{SWAP}}$.  Appendix~\ref{appendix:proofofswap} proves that $\chi$ at $\mu=1$ is exact for $\mathbb{L}_{\text{SWAP}}$.

Because the $T_6$ QCA and $F_4$ QCA do not have SWAP-isomorphic substring edit maps, we do not have any available proofs for their exactness, and instead rely on simulating $\chi$ for these models. The results of this are in Fig.~\ref{fig:falseminima}, where for the $T_6$ and $F_4$ QCA models we plot the proportion of states $\ket{b} \in B(\mathcal{H}_d)$ such that $\chi(\ket{b},\mu) \neq \min(G_{\ket{b}})$ for each $n \in [k, 17]$. 

First, notice Fig.~\ref{fig:falseminima}(a), which shows search failures for the $T_6$ QCA model; while $\mu=1$ causes most searches to fail as $n$ increases, $\mu=2$ causes \textit{every} search to succeed for every $n$ in the domain. Thus, it we can assume that the heuristic for Alg.~\ref{alg:greedypath} is accurate for this model and that $\chi(\ket{b},\mu=2) = \min(G_{\ket{b}}),$ $\forall \ket{b} \in B$.

Next, notice Fig.~\ref{fig:falseminima}(b), which shows search failures for the $F_4$ QCA model; we use $\mu =1,3,5,7,9$, and find that $\chi$ can fail even with high $\mu$. We see that $\mu =1,3$  sees majority failures over the domain, that $\mu=5,7,9$ has bounded failure $\approx 0.05$, and that no $\mu$ completely removes failures.

This metric only determines the reliability of $\chi$, not the reliability of post-selecting with $\chi$. This is because, as discussed previously in Sec.~\ref{subsec:pathfinding}, two searches can fail and obtain the same false minima, $\chi(\ket{\psi_0},\mu) =\chi(\ket{b_f},\mu) \neq \min(G_{\ket{\psi_0}})=\min(G_{\ket{b_f}})$, i.e., two failed searches can still successfully identify that the states share a conserved quantity. Therefore, the search reliability shown in Fig.~\ref{fig:falseminima} should be seen as an \textit{upper-bound} on the error caused by the heuristic. 

Sections~\ref{subsec:exactps} and~\ref{subsec:probps} explore this reliability further in the context of post-selection on noisy quantum data. In the context of post-selection, we find that even small $\mu=3$ is acceptable for the $F_4$ QCA, as the most likely measurement outcomes also share false minima with the initial state. We discuss this observation more with the data that supports it in Sec.~\ref{subsec:probps}.

\begin{figure}[ht]
        \centering
        \includegraphics[width=\linewidth]{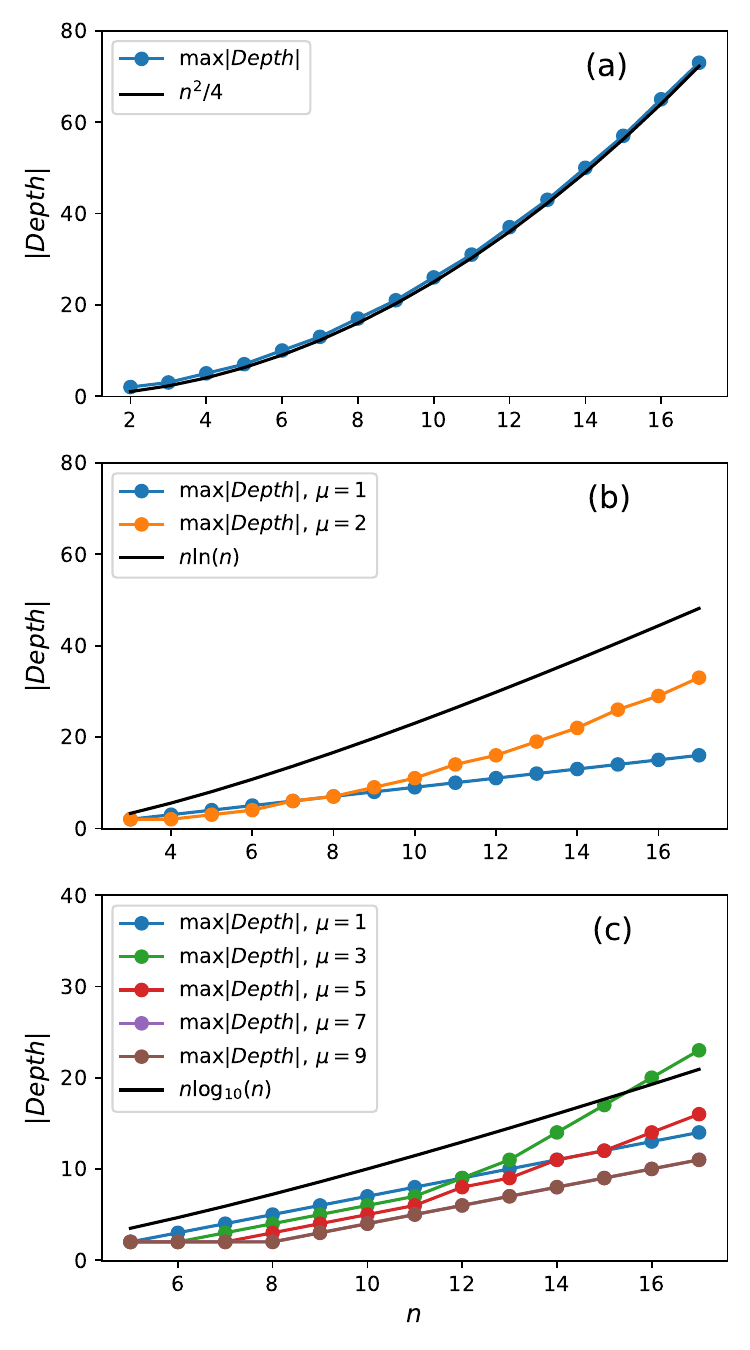}
        \caption{The longest path taken at each $n$ for the models (a) Heisenberg-XXX (b) $T_6$ QCA (c) $F_4$ QCA.}
        \label{fig:pathlength}
\end{figure}

\subsubsection{Search depth}\label{subsubsec:pathlengthdata} 
The time performance of Alg.~\ref{alg:greedypath}  is heavily dependent on the length of the path taken. The asymptotic worst-case of the search depth is system-dependent, being identified as $n^2/4$ for the Heisenberg-XXX model in Appendix~\ref{appendix:pathlengthcalcalc}, and closer to $O(n\ln(n))$ for the $T_6$ QCA and $O(n\log_{10}(n))$ for the $F_4$ QCA by observing the largest search depth for the algorithm at small $n$, shown in Fig.~\ref{fig:pathlength}. In this figure we find the largest path length at each system size for each relevant $\mu$. For the $T_6$ model, shown in Fig.~\ref{fig:pathlength}(a), search depth becomes longer as the parameter $\mu$ increases, indicating that false minima are avoided. For the $F_4$ model, shown in Fig.~\ref{fig:pathlength}(b), the search depth becomes shorter as $\mu$ increases; this is because higher $\mu$ means more edges are traversed at each step in the search, and thus less depth is needed in the search overall.

These path lengths, combined with the $\mu$ found in Sec.~\ref{subsubsec:pathaccuracydata}, results in the computational complexity of Alg.~\ref{alg:greedypath} on a single state being: Heisenberg-XXX is $O(mn^2)$, $T_6$ QCA is $O(m^2n\ln(n))$, and $F_4$ QCA is $O(m^\mu n\log_{10}(n))$. The $\mu$ exponent is left in the $F_4$ QCA computational complexity because, as shown in Sec.~\ref{subsubsec:pathaccuracydata}, we find no value of $\mu$ which guarantees that a search finds the minimal element, and thus expect $\mu$ to vary depending on the symmetry-protected subspace occupied for the simulation. As we can see in Fig.~\ref{fig:falseminima}(b), $\mu \geq 5$ is mostly sufficient.

\begin{figure*}[ht!]
        \centering
        \includegraphics[width=\linewidth]{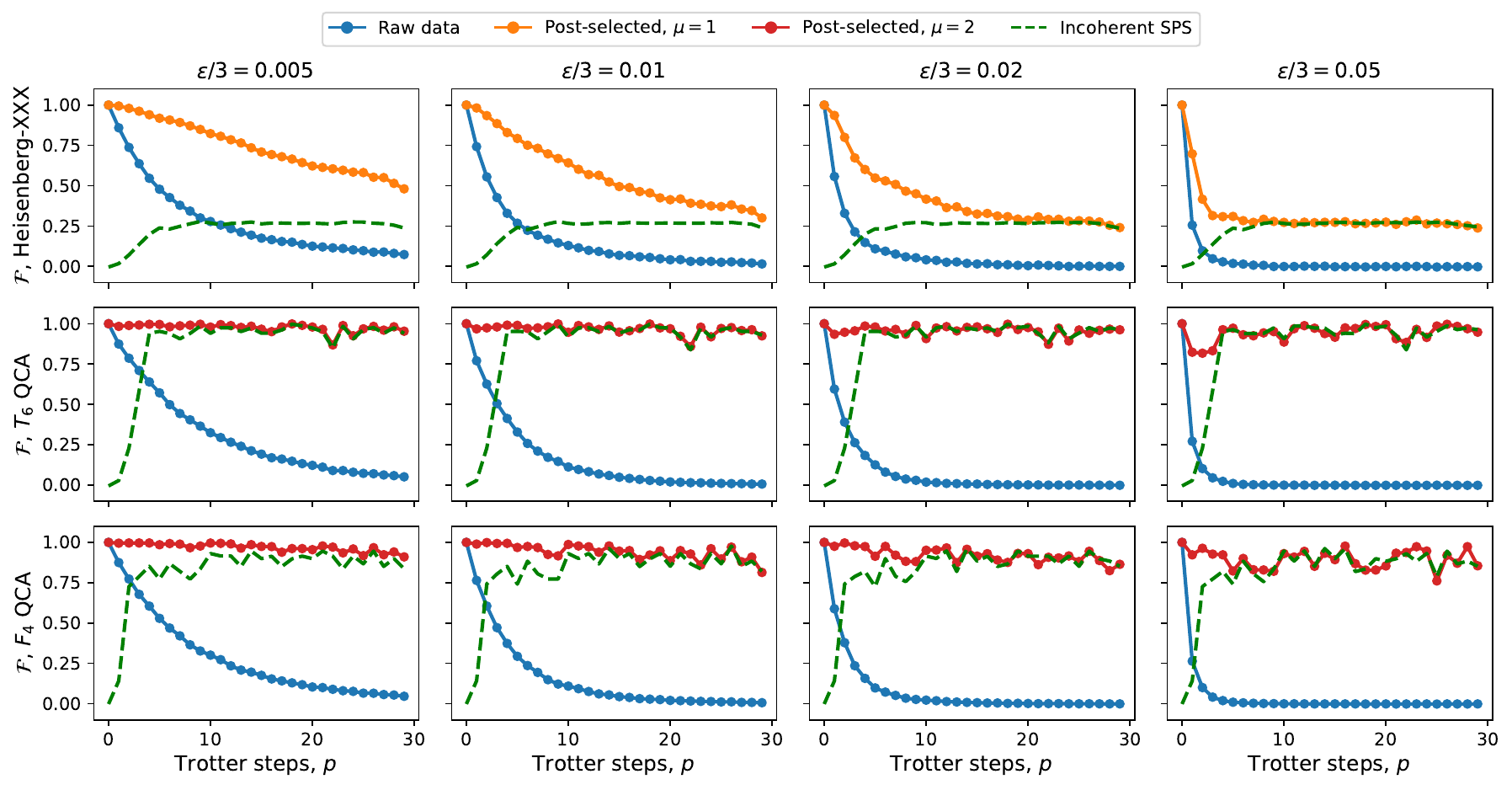}
        \caption{The fidelities $\mathcal{F}(P_{\text{sim}}(p), P_{\text{ideal}}(p))$, $P_{\text{sim}}=P_{\text{raw}}, P_{\text{ps},\mu=1}, P_{\text{ps},\mu=2}, P_{\text{ISPS}}$ of each model (rows) at errors $\epsilon/3=0.005, 0.01, 0.02, 0.05$ (columns). There are $30,000$ measurements at each Trotter step. The $\mu$ used is the smallest required to reach perfect shared subspace verification under Alg.~\ref{alg:greedypath}. This figure shows that with post-selection, model fidelity remains above random noise at least as long as $p=2*n$. }
        \label{fig:allps}
\end{figure*}

\section{Post-Selection in Emulated Noisy Simulations}\label{sec:psdemo}
To demonstrate the power of our algorithm, we perform automated post-selection on simulated quantum computations with noise. We perform discretized time evolutions with our exemplar systems, the Heisenberg-XXX, $T_6$ QCA, and $F_4$ QCA models, and show that our post-selection methods restructure noisy measurement probability distributions to be closer to the one found in an ideal simulation. We first run the simulation without errors via $U(t)\ket{\psi_0}$, constructing the measurement distribution at each time step. Next, we repeated the simulations with depolarizing noise injected after each layer of gates with independent probability $\epsilon/3=0.005, 0.01, 0.02, 0.05$ of $X_i, Y_i, Z_i$ on each qubit $i$. See Appendix~\ref{appendix:simnoisemodel} for more noise model details. We use Kullback-Liebler divergence~\cite{kullback1951information} to quantify the distance between the ideal and noisy measurement distributions constructed at each measurement layer, with and without post-selection using Alg.~\ref{alg:greedypath} and the symmetry-check in Eq.~\eqref{eq:ps_with_path}.

We observe a significant increase in accuracy of the data when using our post-selection methods without adding circuit runs, even in simulation subspaces which have imperfect pathfinding.

\subsection{Methods}\label{subsec:methods}
We simulate discrete time evolution, and for each discrete time step a sequence of measurements is used to construct a probability distribution. Let $P(p)$ be the measurement distribution of the wavefunction after $p$ Trotter steps, such that $P(p)\approx || \ket{\psi(p)} ||^2$ up to shot noise. We compare an ideal simulation measurement distribution $P_{\text{ideal}}(p)$, a noisy simulation without post-selection $P_{\text{raw}}(p)$, and a noisy simulation with post-selection using Alg.~\ref{alg:greedypath} at a given $\mu$,  $P_{\text{ps}, \mu}(p)$. Let $P(p)$ be a normalized probability distribution constructed by a sequence of $M$ measurements after $p$ Trotter steps, and $P(p,b)\approx ||\bra{b}\ket{\psi(p)}||^2$, which is the probability amplitude of state $\ket{b}$ at that layer $p$, up to shot noise. We use Kullback-Liebler divergence~\cite{kullback1951information},
\begin{equation}\label{eqn:kldiv}
\mathcal{D}(P(p), Q(p)) \equiv \sum_{b \in B(\mathcal{H}_d)} P(p,b) \ln (\frac{P(p,b)}{Q(p,b)}),
\end{equation}
 which measures the distance between two probability distributions, and use the following equation from~\cite{PhysRevResearch.2.043042} as a simulation's fidelity,
 \begin{equation}\label{eqn:errormetric}
 \mathcal{F}(P_{\text{sim}}(p), P_{\text{ideal}}(p)) \equiv 1 - \frac{\mathcal{D}(P_{\text{sim}}(p), P_{\text{ideal}}(p))} {\mathcal{D}(P_{\text{IRN}}, P_{\text{ideal}}(p))},
 \end{equation}
where $P_{\text{IRN}}$ is the \textbf{i}ncoherent \textbf{r}andom \textbf{n}oise probability distribution
\begin{equation}
P_{\text{IRN}}(b) \equiv \frac{1}{d} \: \: \forall \: \: \ket{b} \in B(\mathcal{H}_d) 
\end{equation}
which we expect to measure once noise has proliferated in the computation. Equation~\eqref{eqn:errormetric} will return 1 if $P_{\text{sim}}= P_{\text{ideal}}$, and 0 if $P_{\text{sim}}= P_{\text{IRN}}$. We compute Eq.~\eqref{eqn:errormetric} at each measurement layer in the computation, with $P_{\text{sim}}$ being either the raw data $P_{\text{raw}}$ or the post-selected data $P_{\text{ps}}$.

We post-select by removing measurement results that fail under Alg.~\ref{alg:greedypath} and Eq.~\eqref{eq:ps_with_path} and re-normalizing the probability distribution. If $P_\text{raw}(p)$ is the measurement distribution with noise at step $p$ of $M$ measurements, $P_{\text{ps}}(p)$ is the post-selected measurement distribution, given by checking each state in $P_\text{raw}(p)$ in Alg.~\ref{alg:greedypath}, at the same Trotter-step $p$ with measurements $M_\text{kept} \leq M$.

\begin{figure*}[ht]
        \centering
        \includegraphics[width=\linewidth]{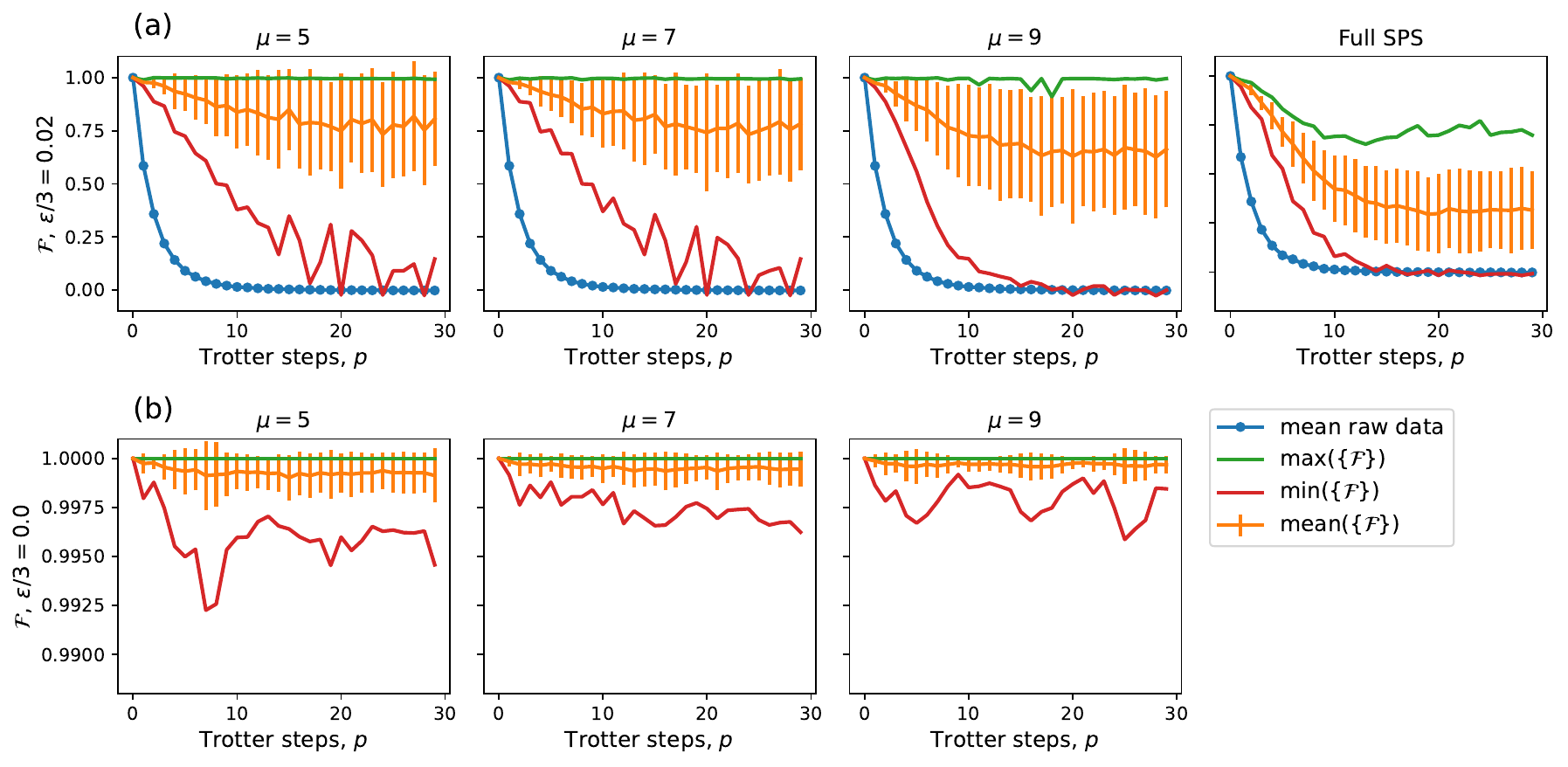}
        \caption{The maximum, minimum, mean, and standard deviation of the set of fidelities, $\{\mathcal{F}\}$, generated by simulating the $F_4$ QCA for an initial condition in each of its incomplete subspaces at $\mu=9$ with post-selected distributions $P_{\text{sim}} =P_{\text{ps}, \mu=5}, P_{\text{ps}, \mu=7},P_{\text{ps}, \mu=9}, P_{\text{SPS}}$. Each Trotter-step has $M=10000$ measurements. The error rate is (a) $\epsilon/3=0.02$ and (b) $\epsilon/3=0.0$, or an ideal simulation.}
        \label{fig:badps}
\end{figure*}

\subsection{Post-selection with perfect symmetry-protected subspaces}\label{subsec:exactps}
We run a 15 qubit simulation of each of the three exemplar models, using initial conditions seen in the literature, the anti-ferromagnetic state $\ket{\psi_0} = \ket{1010\hdots101}$ for the Heisenberg-XXX model, a single bit-flip on the middle qubit $\ket{\psi_0} = \ket{0\hdots1\hdots0}$ for the $T_6$ QCA~\cite{Jones2022,Hillberry2021}, and two bit-flips neighboring the center qubit $\ket{\psi_0} = \ket{0\hdots101\hdots0}$ for the $F_4$ QCA~\cite{Hillberry2021}. No false minima are seen past $\mu=2$ in any of the simulations for these initial conditions, meaning post-selection functions effectively. The data uses $M=30,000$ measurements at each measurement layer and does not add any measurements to replace error-victim circuit runs. 

Results can be seen in Fig.~\ref{fig:allps}. We observe that post-selected data approaches incoherent noise in the SPS of that initial state $P_{\text{ISPS}}$, defined as 
\begin{equation}
P_{\text{ISPS}}(b) \equiv
\begin{cases}
\frac{1}{|G|} &\:  \ket{b} \in G\\
0 &\:  \ket{b} \notin G\\
\end{cases}
\end{equation}
Because our simulation's raw data converges to random noise, which is a uniform probability distribution over the Hilbert space, and post-selection only removes measurement outcomes, we should expect that the result of post-selection is still a uniform probability distribution, just over the symmetry-protected subspace instead. In other words, we should observe that 
\begin{equation}
\lim_{p\rightarrow\infty}\mathcal{F}(P_{\text{ps}}(p), P_{\text{ideal}}(p)) \approx \mathcal{F}(P_{\text{ISPS}}(p), P_{\text{ideal}}(p))
\end{equation}
given enough measurements.

\subsection{Post-selection with incorrect symmetry-protected subspaces}\label{subsec:probps}
We show that even in subspaces of the $F_4$ QCA model where Alg.~\ref{alg:greedypath} has failures at high $\mu$, which is equivalent to false-positive mislabeling some basis vectors as \textit{outside} the symmetry-protected subspace, our methods mitigate more errors than they introduce and thus still show promise for NISQ quantum simulations. 

We found 18 symmetry-protected subspaces of the $F_4$ QCA that had false minima at $n=15$ with $\mu\geq9$ (henceforth called incomplete subspaces) and ran simulations to $p=29$ with $M=10000$ measurements at each cycle, initialized in one state from each rocky subspace. See Appendix~\ref{appendix:f4_sim_initconds} for a list of these initial states. This way, Alg.~\ref{alg:greedypath} is highly likely to meet false minima and reject measurements that were in the symmetry-protected subspace of the initial state. We calculate $\mathcal{F}(P_{\text{ps},\mu}(p), P_{\text{ideal}}(p))$ for each simulation, and run our post-selection method with $\mu =5,7,9$. See Fig.~\ref{fig:badps}(a) for the set of fidelities $\{\mathcal{F}\}$ at $\epsilon/3=0.02$, compared to post-selection with the ``Full SPS'', which is post-selection using the full symmetry-protected subspace, instead of the heuristic; in the average case our post-selection keeps the simulation data above incoherent noise, even with lower $\mu$. The worst-case fidelity (in red) drops to the fidelity of incoherent noise, which indicates that our methods only offer improvement on particular simulations. However, the mean (orange) and standard deviation (orange bars) remain reliably higher than incoherent noise, with the maximum fidelity keeping close to the ideal simulation. 

The other important observation is that our post-selection, which is using an inaccurate symmetry-protected subspace, obtains a higher mean and maximum fidelity than post-selection with a well-defined symmetry-protected subspace, with fidelity lowering as $\mu$ increases. We theorize the reason for this is as follows: if two states share a false minima with $\chi$, they are likely close to each other in the string interaction graph $D_\mathbb{L}$, and thus have a higher probability of transitions between each other in the simulation. On the other hand, if a measured state does \textit{not} share a false minima with the initial state, it is less likely to be observed in a simulation, but a noisy simulation will artificially amplify the probability of that state. The result is that post-selection with the true SPS will accept these artificially amplified results, where our false SPS will coincidentally reject states that can possibly have this happen. 

We also examine our post-selection in a simulation absent any error; this is to see how much a false rejection can degrade a perfect computation. See Fig.~\ref{fig:badps}(b) for post-selection with an error-free simulation. It can be seen that most fidelities remain above $\approx0.999$, meaning most of the time this error-prone post-selection method will have minimal impact even on an ideal simulation. The worst-case degraded fidelity is $\approx 0.992$.

\section{Discussion and conclusion}
We introduce a graph theory interpretation of quantum time evolution, which provides a theoretical framework through which symmetry-protected subspaces can be constructed via transitive closure. We identify that these invariant subspaces are an operator-free method for characterizing symmetry in the system, by indicating a conserved quantity of an initial state as it is manipulated by the dynamical system without needing to explicitly identify it. Along these lines, our approach complements recent work in open quantum systems where symmetry-protected density matrices were constructed without the knowledge of explicit symmetry operators \cite{thingna2021degenerated}.

We observe that a symmetry-protected subspace can be used to provide a smaller computational space, post-select noisy quantum simulation data, or be analyzed to deduce a symmetry operator. We  identify post-selection as a pertinent application and introduce two main classical algorithms that elucidate the features of a quantum system's symmetry-protected subspaces. These algorithms employ a \textit{basis string edit map}, which is an efficient construction to provide the local dynamics of an operator by focusing on the presence or absence of basis vectors through its action.

The first algorithm uses transitive closure, calculated with breadth-first search of these basis string edit maps, to enumerate every state within the symmetry-protected subspace of the initial state. Because post-selection with this subspace requires constructing the entire SPS, which can still be exponentially difficult, we provide a second, polynomial-scaling algorithm. The second algorithm attempts to find the smallest basis state (by binary-encoded integer) in an SPS by following a locally-optimal heuristic, and establishing that a path of basis string edit maps exists between an initial and final state if their respective searches  collide at any elements.

We conclude by demonstrating post-selection using our second algorithm, which shows that even when the raw simulation data degrades to incoherent noise, our methods effectively recover a probability distribution closer to an ideal computation. Our methods are compatible with any other error-mitigation technique compatible with post-selection, such as zero noise extrapolation, and thus present a further addition to a growing array of techniques to improve noisy quantum computation~\cite{endo2021hybrid}.

We identify a few obvious extensions to this project. First, if Alg.~\ref{alg:makesps} concludes with certain classical resources, this implies that the wavefunction can be stored with constant overhead on those same resources; this could shrink the computational memory requirements from the quantum regime to the classical regime. Second, we speculate that more reliable algorithms than Alg.~\ref{alg:greedypath} may fulfill the same function; a string-matching algorithm similar to Needleman-Wunsch~\cite{NEEDLEMAN1970443} with the basis string edit maps, instead of insert/delete/shift edits, may fulfill this function.

This work also presents interesting results on post-selection that are worth further exploration. The first is the convergence to incoherent noise \emph{within} a symmetry-protected subspace; this shows that, as expected, post-selection will not converge to an ideal computation, regardless of the additional circuit runs to replace error-victim runs detected, as error will still have proliferated within the the SPS. However, such incoherent  distributions still have significant complexity and nontrivial structure (see e.g.,~\cite{Jones2022}), making them particularly interesting in cases where the underlying symmetry is not analytically known. The second is the observation that mis-identifying measurements within the SPS paradoxically results in a more accurate computation than what results from full knowledge of the SPS. The exact cause of this is up to speculation, and is likely a problem-dependent effect, but is interesting in its own right.

Finally, since the algorithms presented herein are specified with respect to a particular simulation or measurement basis, our methods are currently constrained to discovering subspaces that are protected either by a single symmetry generator in that basis or by a collection of Abelian generators, which are diagonal (or have diagonal representations) in the chosen basis. However, constraining many-body dynamics by sets of non-Abelian generators often yields rich physics as has been noted with respect to non-Abelian thermal states \cite{yunger2016microcanonical}, bipartite entanglement entropy growth \cite{majidy2023non}, eigenstate thermalization \cite{murthy2022non}, and thermalization in finite-size quantum simulators \cite{halpern2020noncommuting, kranzl2022experimental}. As such, extending our subspace detection framework to subspaces that are protected by multiple non-Abelian generators would be a fruitful direction for future research.

Code and data are available upon reasonable request to the corresponding authors.

\section{Acknowledgements}

This material is based upon work supported by the U.S. Department of Energy, Office of Science, National Quantum Information Science Research Centers, Superconducting Quantum Materials and Systems Center (SQMS) under contract number DE-AC02-07CH11359, and by National Science Foundation grant PHY-1653820. 
This work was also authored in part by the National Renewable Energy Laboratory (NREL), operated by Alliance for Sustainable
Energy, LLC, for the U.S. Department of Energy (DOE) under Contract No. DE-AC36- 08GO28308. This work
was supported by the Laboratory Directed Research and Development (LDRD) Program at NREL. The views expressed
in the article do not necessarily represent the views of the DOE or the U.S. Government. The U.S. Government retains
and the publisher, by accepting the article for publication, acknowledges that the U.S. Government retains a nonexclusive,
paid-up, irrevocable, worldwide license to publish or reproduce the published form of this work, or allow others
to do so, for U.S. Government purposes.

\renewcommand{\appendixname}{APPENDIX}
\appendix

\appsection{PROOF OF ACCEPTABLE MAXIMAL COVERAGE IN SYMMETRY-PROTECTED SUBSPACES}\label{appendix:proofdestint}

Take the operator $P_G = \sum _{b\in G_{\ket{\psi_0}}} \ket{b}\bra{b}$ to be the projection operator for the symmetry-protected subspace of $\ket{\psi_0}$:
\begin{equation}
[P_{G_{\ket{\psi_0}}}, U] = 0
\end{equation}
and suppose it has been constructed such that for each $\ket{b} \in G_{\ket{\psi_0}}$, $\exists\, t\; s.t. \; \bra{b}U^t\ket{\psi_0} \neq 0$. Now suppose that the basis vector $\ket{k}$ is never seen by time evolution of $U$, such that 
\begin{equation}\label{eqn:appendixnonspsstate}
\bra{k}U(t)\ket{\psi_0}=0\; \forall t.
\end{equation}
If we include $\ket{k}$ and the symmetry-protected subspace of $\ket{k}$ in our protected subspace, such that $G' = G_{\ket{\psi_0}} \cup G_{\ket{k}}$, then $P_{G'} = P_{G_{\ket{\psi_0}}} + \sum_{j\in G_{\ket{k}}} \ket{j}\bra{j}$, and we get
\begin{align}
[P_{G'}, U] &= [P_{G_{\ket{\psi_0}}} + P_{G_{\ket{k}}}, U]\\
        &= [P_{G_{\ket{\psi_0}}}, U] + [P_{G_{\ket{k}}}, U]\\
        &= 0
\end{align}
by the Def.~\ref{def:sps} of a symmetry-protected subspace for both $G_{\ket{\psi_0}}$ and $G_{\ket{k}}$. This shows that if a state which is never seen by the time evolution is included in the final symmetry-protected subspace, the projection operator still commutes and thus the subspace still follows Def.~\ref{def:sps}.

\appsection{Proof of Theorem~\ref{thm:spstheorem}}\label{appendix:prooftc}

We take $\ket{\psi_0}\in B$ to be a product state. Due to linearity, generalizing the proof to superposition states is trivial, given that one interprets $G_{\ket{\psi_0}}$ as being the union of all symmetry-protected subspaces of each of which corresponds to one (or more) basis states in the superposition. For notational simplicity, and without loss of generality, we also assume an implicit operator ordering $\mathcal{O}_{op}$. Consider the transition amplitude

\begin{equation}\label{eq:appx_amplitude}
\begin{split}
\bra{b_f}U(t)\ket{\psi_0}&=\bra{b_f}\prod_{j=1}^p\prod_{i=1}^m U_i(\tau_j)\ket{\psi_0}\\
&=\bra{b_f}\prod_{j=1}^p\prod_{i=1}^m \sum_{\ket{b_{j,i}}\in B} \ket{b_{j,i}}\bra{b_{j,i}} U_i(\tau_j)\ket{\psi_0},
\end{split}
\end{equation}
where in the second line of Eq.~\eqref{eq:appx_amplitude} we have inserted the identity operator $\mathds{1}=\sum_{\ket{b_{j,i}}\in B} \ket{b_{j,i}}\bra{b_{j,i}}$ after each $k$-local unitary operator. Without loss of generality, we assume an operator ordering to expand and rearrange the right-hand side of Eq.~\eqref{eq:appx_amplitude},

\begin{equation}\label{eq:amplitude_expansion}
\begin{split}
&\bra{b_f}U(t)\ket{\psi_0}=\\
&=\sum_{\ket{b_{p,m}}\in B} \bra{b_f}\ket{b_{p,m}} \sum_{\ket{b_{p,m-1}}\in B}\bra{b_{p,m}}U_{m}(\tau_p)\ket{b_{p,m-1}} \ldots \\
&\ldots \sum_{\ket{b_{1,1}}\in B} \bra{b_{1,2}}U_2(\tau_1)\ket{b_{1,1}} \bra{b_{1,1}} U_1(\tau_1)\ket{\psi_0}.
\end{split}
\end{equation}
Next, we note that by the definition of the string edit map, Def.~\ref{def:stringeditoperator}, and the presence of the local transition amplitude $\bra{b_{1,1}} U_1(\tau_1)\ket{\psi_0}$ in Eq.~\eqref{eq:amplitude_expansion}, we can re-write

\begin{equation}\label{eq:resum}
\begin{split}
& \sum_{\ket{b_{1,1}}\in B} \bra{b_{1,2}}U_2(\tau_1)\ket{b_{1,1}} \bra{b_{1,1}} U_1(\tau_1)\ket{\psi_0} = \\
=&\sum_{\ket{b_{1,1}}\in \mathcal{L}_{U_1}(\ket{\psi_0})} \bra{b_{1,2}}U_2(\tau_1)\ket{b_{1,1}} \bra{b_{1,1}} U_1(\tau_1)\ket{\psi_0} \\
=&\sum_{\ket{b_{1,1}}\in T^1_{\ket{\psi_0}}} \bra{b_{1,2}}U_2(\tau_1)\ket{b_{1,1}} \bra{b_{1,1}} U_1(\tau_1)\ket{\psi_0},
\end{split}
\end{equation}
where in the third line of Eq.~\eqref{eq:resum} we have used the fact that $\mathcal{L}_{U_1}(\ket{\psi_0})\subseteq T^1_{\ket{\psi_0}}\subseteq B$, which is to say, we can expand the summation to run over $T^1_{\ket{\psi_0}}$ by summing over terms where the local transition amplitude is zero. By the same logic, we can re-write the second summation as
\begin{equation}\label{eq:second_resum}
\begin{split}
\sum_{\ket{b_{1, 2}} \in B} (\ldots) &=\sum_{\ket{b_{1, 2}} \in \mathcal{L}_{U_2}(\ket{b_{1,1}})\: : \: \ket{b_{1,1}}\in \mathcal{L}_{U_1}(\ket{\psi_0})} (\ldots)\\
&= \sum_{\ket{b_{1, 2}} \in T^2_{\ket{\psi_0}}}(\ldots),
\end{split}
\end{equation}
where again, we have simplified notation by noting that $\{\mathcal{L}_{U_2}(\ket{b_{1,1}})\: : \: \ket{b_{1,1}}\in \mathcal{L}_{U_1}(\ket{\psi_0})\} \subseteq T^2_{\ket{\psi_0}}$ and summing over all additional states in $T^2_{\ket{\psi_0}}$ to which the transition amplitude is zero. Generally, one can re-write summation $(j-1)m+i$ as

\begin{equation}\label{eq:arb_resum}
\sum_{\ket{b_{j, i}} \in B} (\ldots) = \sum_{\ket{b_{j, i}} \in T^{(j-1)m+i}_{\ket{\psi_0}}}(\ldots).
\end{equation}
We therefore write the full transition amplitude as 

\begin{equation}\label{eq:final_amplitude_expansion}
\begin{split}
&\bra{b_f}U(t)\ket{\psi_0}=\\
&\sum_{\ket{b_{p,m}}\in T_{\ket{\psi_0}}^{pm}} \bra{b_{f}}\ket{b_{p,m}}\sum_{\ket{b_{p,m-1}}\in T_{\ket{\psi_0}}^{pm-1}}\bra{b_{p,m}}U_{m}(\tau_p)\ket{b_{p,m-1}} \\
&\ldots \sum_{\ket{b_{1,1}}\in T^1_{\ket{\psi_0}}} \bra{b_{1,2}}U_2(\tau_1)\ket{b_{1,1}} \bra{b_{1,1}} U_1(\tau_1)\ket{\psi_0}.
\end{split}
\end{equation}
There are now two cases to consider. Let $T_{\ket{\psi_0}}^{\ell}$ be the iteration at which the stop condition in Eq.~\eqref{eqn:recursionrelation} is activated. Either $T_{\ket{\psi_0}}^{pm} = T_{\ket{\psi_0}}^{\ell} \equiv G_{\ket{\psi_0}}$, or else $T_{\ket{\psi_0}}^{pm} \subset T_{\ket{\psi_0}}^{\ell}$. In either case, $\ket{b_f}\notin G_{\ket{\psi_0}} \implies \ket{b_f}\notin T_{\ket{\psi_0}}^{pm}$, which means that $\bra{b_f}\ket{b_{p,m}}=0\: \: \forall \: \: \ket{b_{p,m}} \in T_{\ket{\psi_0}}^{pm}$ and the full transition amplitude must vanish. \qedsymbol{}

\appsection{Analysis of Algorithm~\ref{alg:greedypath}}\label{appendix:proofofgreedypathfinding}
\subsection{Proof of exactness for $\mathbb{L}_{\text{SWAP}}$ isomorphic systems}\label{appendix:proofofswap}
We will show that for a system $U$ which has $\mathbb{L}_{U} = \mathbb{L}_{\text{SWAP}}$, iterative path creation with Eq.~\eqref{eqn:iterativepathbuild} (which is what Alg.~\ref{alg:greedypath} computes) has only one stop condition when using $\mu=1$, which is the binary encoded minima of the entire symmetry-protected subspace.

Assume our simulation of $n$ qubits has substring edit maps $\mathbb{L}_{\text{SWAP}} \equiv \{\mathcal{L}_{\text{SWAP}_{0,1}},\mathcal{L}_{\text{SWAP}_{1,2}},\hdots,\mathcal{L}_{\text{SWAP}_{n-2,n-1}}\}$. We take $\ket{b_0}$ to be the first state in the path, and it is a bitstring $\ket{b_0} = \ket{\{0,1\}^n}$ such that $\sum_{i=0}^{n-1} a_i^\dagger a_i \ket{b_0} = s\ket{b_0}$. The smallest binary encoding belonging to a state in the symmetry-protected subspace of $\ket{b_0}$, which is the SPS encoding particle conservation symmetry, will belong to the bitstring $\min(G_{\ket{b_0}}) = \ket{1^s0^{n-s}}$. The bits are ordered left-to-right as least-to-most significant. 

The string edit map $\mathcal{L}_{\text{SWAP}_{i,i+1}}$ acting on its local basis set $B_i =\{ \ket{00}, \ket{01}, \ket{10}, \ket{11}\}$ has the following mappings:
\begin{equation}\label{eq:allswapmaps}
\begin{split}
\mathcal{L}_{\text{SWAP}_{i,i+1}}(\ket{00})&= \{\ket{00}\}\\
\mathcal{L}_{\text{SWAP}_{i,i+1}}(\ket{01})&= \{\ket{10},\ket{01}\}\\
\mathcal{L}_{\text{SWAP}_{i,i+1}}(\ket{10})&= \{\ket{01},\ket{10}\}\\
\mathcal{L}_{\text{SWAP}_{i,i+1}}(\ket{11})&= \{\ket{11}\}\\
\end{split}.
\end{equation}
When applied to a $n$-qubit state represented as an $n$-character bitstring $b$, the mapping operates as e.g. 
\begin{equation}\label{eq:mapexampleforproof}
\begin{split}
\mathcal{L}_{\text{SWAP}_{i,i+1}}(&\ket{b[0,i-1];01;b[i+2,n-1]}) \\
&= \ket{b[0,i-1];10;[i+2,n-1]},
\end{split}
\end{equation}
where the ``$;$'' symbol means string concatenation and ``$b[i,j]$'' is the substring of $b$ from characters at index $i$ through $j$ (inclusive).

Given an incomplete search with $\chi(\ket{b_0},\mu=1)$, where $\ket{b_{j}}=\min(T^1_{\ket{b_{j-1}}})$ is the most recent progression in the search, the next state $\ket{b_{j+1}}$ in the search is the smallest bitstring in $\mathbb{L}_{\text{SWAP}} (\ket{b_j})$, given by $\ket{b_{j+1}}=\min(T^1_{\ket{b_j}})$; this will be $\mathcal{L}_{\text{SWAP}_{i,i+1}}(\ket{b_{j}})$, where $\ket{b_j[i+1]}$ is the \textit{most} significant bit equal to $\ket{1}$ such that $\ket{b_j[i]}=\ket{0}$. 

In other words, our heuristic dictates that we will only choose $\mathcal{L}_{\text{SWAP}_{i,i+1}}(\ket{b_j})$ when $\ket{b_j[i,i+1]} = \ket{01}$. Therefore, Alg.~\ref{alg:greedypath} will continue until there is no substring ``01'' in $b_j$. For a bitstring, the only configuration of bits that meets this condition is $b=1^{s}0^{n-s}$, which is the state $\min(G_{\ket{b_0}})$. Therefore, if Alg.~\ref{alg:greedypath} halts, it must have found the minimum binary encoded state for that symmetry-protected subspace. \qedsymbol{}

\subsection{Max depth of search in $\mathbb{L}_{\text{SWAP}}$}\label{appendix:pathlengthcalcalc}
 
Here we calculate the worst-case search depth for Alg.~\ref{alg:greedypath} when using string edit maps $\mathbb{L}_{\text{SWAP}}$. Each time a new state $\ket{b_{j+1}}$ is traversed to with $\ket{b_{j+1}} = \min(T^1_{\ket{b_j}})$ from Eq.~\eqref{eqn:iterativepathbuild}, a ``1'' bit is swapped with a left-neighboring ``0'' bit, assuming bits are least-to-most significant left-to-right. We count each of these swaps as 1 step in the search, and this subsection will show an analytical form for the \textit{most} steps ever required for an arbitrary bitstring $\ket{b_0}$ of length $n$ to reach the stop condition $\ket{1^s0^{n-s}}$, as defined at the beginning of Appendix~\ref{appendix:proofofswap}. That section also shows that for the SWAP-isomorphic model we always get $\chi(\ket{b_0}, 1) = \ket{1^s0^{n-s}}$. Note that, because $\mu=1$ is used, the path length of edges in $D_\mathbb{L}$ between $\ket{b_0}$ and $\ket{1^s0^{n-s}}$ is equal to the search depth, because $|Path| = \mu\times|Depth| = |Depth|$.

Suppose we have a symmetry-protected subspace of $s$ many ``1'' bits and $n-s$ many ``0'' bits. The target state $\min(G_{\ket{b_0}})$ for a search with $\chi(\ket{b_0},1)$ is $\min(G_{\ket{b_0}}) = \ket{1^s0^{n-s}}$. The longest path will start with $\ket{b_0} = \ket{0^{n-s}1^s}$, since every other state is fewer SWAPs away from $\min(G_{\ket{b_0}})$. 

Each $\ket{1}$ state in $\ket{b_0}$ requires $n-s$ swaps to get to its position in $\ket{1^s0^{n-s}}$; therefore, in a system of $s$ ``$\ket{1}$'' states, the search depth obeys the summation
\begin{equation}
|Depth| = \sum_{i=1}^s n-s,
\end{equation}
which simplifies to 
\begin{equation}
|Depth| = s(n-s).
\end{equation}
This quantity is maximized at $s=n/2$, which results in $\max(|Depth|) = \frac{n^2}{4}$.

\appsection{Simulation Information}
Our simulations all follow the same template: we take an initial state $\ket{\psi_0}$, which is a $Z$ basis vector prepared by single-qubit Pauli-$X$ gates. In our ideal simulations, we follow this with $p\in [0, 29]$ gate layers of gates to encode the dynamical system, followed by a parallel-readout of all qubits in the $Z$ basis with $M$ measurements. This is shown in the circuit diagram below, where $\ket{\psi_0}$ is a computational basis state prepared by a set of $X$ gates, each $U(\tau_j)$ is a single layer of the Trotterized circuit, at layer $j$, with parameters $\tau_j$. Once the system has been simulated up to $t$-layers, we measure $M$ times, and repeat for each $p \in [0,29]$: \\

\includegraphics[width=0.8\linewidth]{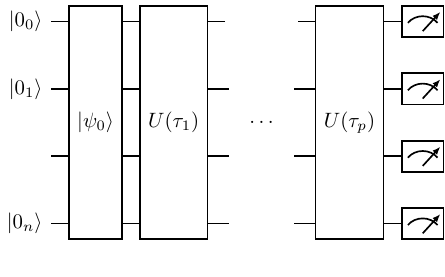}.

For a simulation with noise, we also added an ``error layer'' $E$ after each $U(\tau_j)$, such that $U\rightarrow UE$. It is detailed below in Appendix~\ref{appendix:simnoisemodel}.

\subsection{Simulation noise model}\label{appendix:simnoisemodel}
We use symmetric single-qubit depolarizing noise with probability $\epsilon/3$ after each Trotter layer in the circuit. The circuit diagram for this process at an arbitrary Trotter step $j$ is \\

\includegraphics[width=0.4\linewidth]{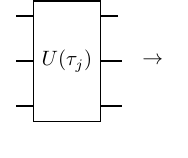}\\
\includegraphics[width=0.8\linewidth]{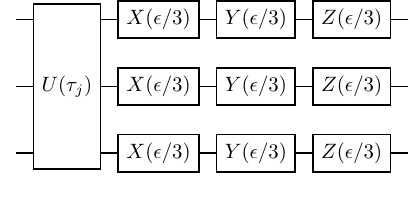},

\noindent where a gate $X(\epsilon/3)$ is an $X$ gate applied with probability $\epsilon/3$ and identity $\mathds{1}$ applied with probability $1-\epsilon/3$. In practice, we simulate this with density matrices through the following steps: if we have the wavefunction's density matrix $\rho$ from the last application of the Trotterized circuit layer $U(\tau_j)$, we apply the symmetric depolarizing noise with the equation
\begin{equation}\label{eqn:depolarizing_noise}
\rho \rightarrow (1-\epsilon)\rho + \frac{\epsilon}{3}X_i\rho X_i + \frac{\epsilon}{3}Y_i\rho Y_i + \frac{\epsilon}{3}Z_i\rho Z_i
\end{equation} on each qubit $i$. As stated in Sec.~\ref{sec:psdemo}, we use $\epsilon/3=0.005, 0.01, 0.02, 0.05$.

\subsection{Heisenberg-XXX}
Here we give the simulation details for the Heisenberg-XXX model, with Hamiltonian given in Sec.~\ref{subsubsec:heisenbergXXX} and simulation results in Sec.~\ref{subsec:exactps}.

\subsubsection{Circuit implementation} 
The $k$-local unitary used ($k=2$) is $U_{i,i+1} = \text{iSWAP}_{i,i+1}(\theta) Z_iZ_{i+1}(\theta)$, or \\

\includegraphics[width=.4\linewidth]{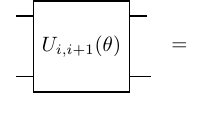}\\
\includegraphics[width=.9\linewidth]{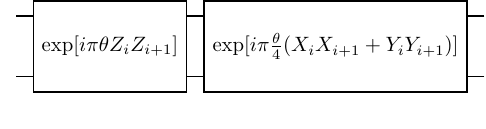},

\noindent where $i$ is the complex coefficient when multiplied and the qubit when indexed, using $\theta=0.1$. We then do an even- and odd-layered Trotterization for one discrete layer of the simulation. This Trotterization on $n=5$ is \\

\includegraphics[width=0.8\linewidth]{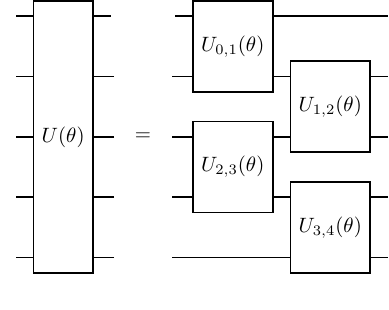}

\noindent for a single discrete time-step. We use the initial condition $\ket{\psi_0} = \ket{101010101010101}$.

\subsubsection{String edit maps}
The string edit maps $\mathbb{L}_{\text{HeisXXX}}$ for this simulation are SWAP-isomorphic, with each local string edit map $\mathcal{L}_{\text{HeisXXX}_{i,i+1}}$ having the behavior:
\begin{equation}\label{eq:allheisxxxmaps}
\begin{split}
\mathcal{L}_{\text{HeisXXX}_{i,i+1}}(\ket{00})&= \{\ket{00}\}\\
\mathcal{L}_{\text{HeisXXX}_{i,i+1}}(\ket{01})&= \{\ket{01},\ket{10}\}\\
\mathcal{L}_{\text{HeisXXX}_{i,i+1}}(\ket{10})&= \{\ket{01},\ket{10}\}\\
\mathcal{L}_{\text{HeisXXX}_{i,i+1}}(\ket{11})&= \{\ket{11}\}\\
\end{split}
\end{equation}
Each string edit (row in Eq.~\eqref{eq:allheisxxxmaps}) obviously conserves $S = \sum_{j=i}^{i+1} Z_j$.

\subsection{$T_6$ quantum cellular automata}
Here we give the simulation details of the $T_6$ quantum cellular automata, with model details in Sec.~\ref{subsubsec:t6qca} and simulation results in Sec.~\ref{subsec:exactps}.

The $k$-local unitary used, ($k=3$), is $U_{i} = P^{(1)}_{i-1}H_iP^{(0)}_{i+1} + P^{(0)}_{i-1}H_iP^{(1)}_{i+1} + P^{(1)}_{i-1}\mathds{1}_iP^{(1)}_{i+1} + P^{(0)}_{i-1}\mathds{1}_iP^{(0)}_{i+1} = \sum_{\alpha,\beta=0}^1 P_{i-1}^{(\alpha)} \otimes (H_i)^{\delta_{\alpha+\beta,1}} \otimes P_{i+1}^{(\beta)}$, where $H_i$ is the Hadamard gate. This unitary is implemented with the circuit\\


\includegraphics[width=0.55\linewidth]{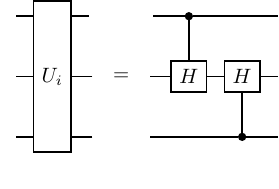}.

\noindent We implement this as a global unitary, which is applied once at each time-step in the simulation by doing a parallel layer of all gates $U_{i}$ where $i$ is even, then a layer of all gates where $i$ is odd. 


\includegraphics[width=0.75\linewidth]{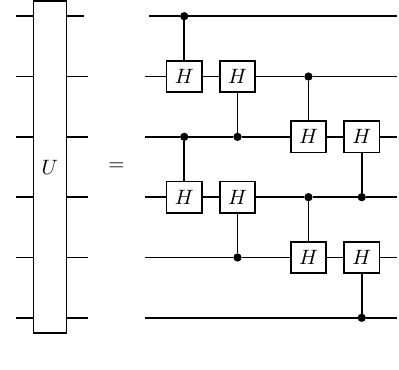}.

\noindent We only use the initial condition $\ket{\psi_0}=\ket{000000010000000}$

\subsubsection{String edit maps}
The string edit maps $\mathbb{L}_{T_6}$ conserve the domain-wall, and each local string edit map $\mathcal{L}_{T_6{i-1,i,i+1}}$ has the behavior
\begin{equation}\label{eq:allt6maps}
\begin{split}
\mathcal{L}_{T_6{i-1,i,i+1}}(\ket{000})&= \{\ket{00}\}\\
\mathcal{L}_{T_6{i-1,i,i+1}}(\ket{001})&= \{\ket{001},\ket{011}\}\\
\mathcal{L}_{T_6{i-1,i,i+1}}(\ket{010})&= \{\ket{010}\}\\
\mathcal{L}_{T_6{i-1,i,i+1}}(\ket{011})&= \{\ket{001},\ket{011}\}\\
\mathcal{L}_{T_6{i-1,i,i+1}}(\ket{100})&= \{\ket{100},\ket{110}\}\\
\mathcal{L}_{T_6{i-1,i,i+1}}(\ket{101})&= \{\ket{101}\}\\
\mathcal{L}_{T_6{i-1,i,i+1}}(\ket{110})&= \{\ket{100},\ket{110}\}\\
\mathcal{L}_{T_6{i-1,i,i+1}}(\ket{111})&= \{\ket{111}\}\\
\end{split}
\end{equation}
which conserves the simulation operator $S=\sum_{j=i-2}^{i+1} Z_jZ_{j+1}$.

\subsection{$F_4$ quantum cellular automata}
Here we give the simulation details of the $F_6$ quantum cellular automata, with model details in Sec.~\ref{subsubsec:f4qca} and simulation results in Sec.~\ref{subsec:exactps} and Sec.~\ref{subsec:probps}. 

The $k$-local unitary used ($k=5$) is built on the local unitary
\begin{equation}
        U_i = \sum_{\alpha, \beta, \gamma, \omega = 0}^1 P_{i-2}^\alpha P_{i-1}^\beta (H_i)^{\delta_{\alpha+\beta+\gamma+\omega,2}} P_{i+1}^\gamma P_{i+2}^{\omega},
\end{equation}
which we implement through quantum gates with the circuit\\


\includegraphics[width=0.95\linewidth]{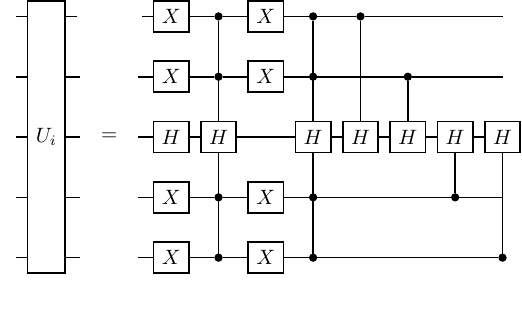}.

\noindent This local circuit is combined into a unitary applied at even discrete time-steps
\begin{equation}\label{eqn:appendixlocalf4even}
    U(\tau_{j=\text{even}}) = \prod_{i = 2, 5, 8, \ldots} U_i(\tau_j) \prod_{i=3,6,9,\ldots} U_i(\tau_j) \prod_{i=4,7,10,\ldots} U_i(\tau_j),
\end{equation}
and a unitary applied at odd discrete time-steps
\begin{equation}\label{eqn:appendixlocalf4odd}
    U(\tau_{j=\text{odd}}) = \prod_{i = 3,6,9, \ldots} U_i(\tau_j) \prod_{i=2,5,8\ldots} U_i(\tau_j) \prod_{i=4,7,10,\ldots} U_i(\tau_j).
\end{equation}
These unitaries on, for example 7 qubits, have the circuit diagrams\\


\includegraphics[width=0.7\linewidth]{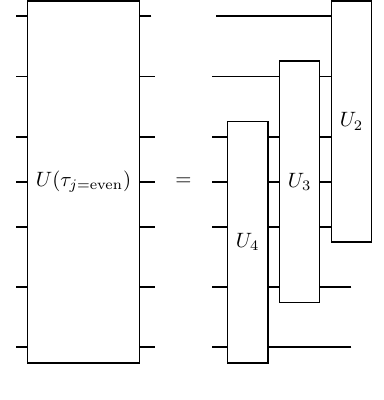}

\noindent and\\


\includegraphics[width=0.7\linewidth]{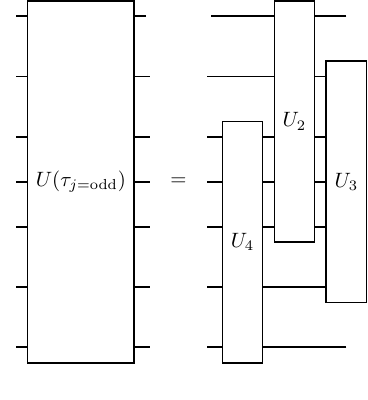}.

\noindent Notice that gates in the same product in Eq.~\eqref{eqn:appendixlocalf4even} and~\eqref{eqn:appendixlocalf4odd}, for example, $U_4$ and $U_7$, commute because only their controls overlap. To get to, for example $t=4$, our circuit would look like\\


\includegraphics[width=0.9\linewidth]{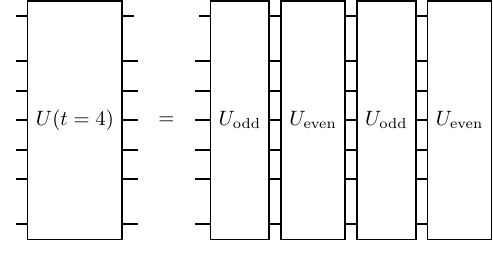}

\subsubsection{String edit maps}
Previously, no symmetry was known for the $F_4$ QCA. Below, we outline the non-identity behavior of each local string edit map $\mathcal{L}_{F_4{i}}$:
\begin{equation}\label{eq:allf4maps}
\begin{split}
\mathcal{L}_{F_4{i}}(\ket{00011})&= \{\ket{00011}, \ket{00111}\}\\
\mathcal{L}_{F_4{i}}(\ket{00111})&= \{\ket{00011}, \ket{00111}\}\\
\mathcal{L}_{F_4{i}}(\ket{01001})&= \{\ket{01001}, \ket{01101}\}\\
\mathcal{L}_{F_4{i}}(\ket{01101})&= \{\ket{01001}, \ket{01101}\}\\
\mathcal{L}_{F_4{i}}(\ket{01010})&= \{\ket{01010},\ket{01110}\}\\
\mathcal{L}_{F_4{i}}(\ket{01110})&= \{\ket{01010},\ket{01110}\}\\
\mathcal{L}_{F_4{i}}(\ket{11000})&= \{\ket{11000},\ket{11100}\}\\
\mathcal{L}_{F_4{i}}(\ket{11100})&= \{\ket{11000},\ket{11100}\}\\
\mathcal{L}_{F_4{i}}(\ket{10010})&= \{\ket{10010},\ket{10110}\}\\
\mathcal{L}_{F_4{i}}(\ket{10110})&= \{\ket{10010},\ket{10110}\}\\
\mathcal{L}_{F_4{i}}(\ket{10001})&= \{\ket{10001},\ket{10101}\}\\
\mathcal{L}_{F_4{i}}(\ket{10101})&= \{\ket{10001},\ket{10101}\}.\\
\end{split}
\end{equation}
Any state $\ket{b}$ which does \textit{not} appear above will have the behavior $\mathcal{L}_{F_4i}(\ket{b}) = \{\ket{b}\}$. It should also be understood that $i$ is the middle qubit in the ordering above. 

\subsubsection{Initial conditions for simulations}\label{appendix:f4_sim_initconds}

We use a variety of initial conditions to test the imperfect path finding of Alg.~\ref{alg:greedypath}. Our initial condition $\ket{\psi_0} = \ket{000000101000000}$, tested in Sec.~\ref{subsec:exactps}, has no false rejections (in other words perfect path finding) with $\mu = 2$.

We also gather the set of symmetry-protected subspaces which have failed paths at $\mu=9$, their average, best, and worst fidelities are shown in Sec.~\ref{subsec:probps}. We chose an arbitrary state from each of these subspaces to be the initial condition; the states are $\ket{\psi_0} \in \{\ket{001000011111111}, \ket{011100000001000}, \\ 
\ket{011110101011111}, \ket{011110000010001}, \ket{011101000001000},\\
\ket{011000000101010}, \ket{001000011111111}, \ket{011100000001000},\\
\ket{011110101011111}, \ket{011110000010001}, \ket{011101000001000},\\
\ket{011000000101010}, \ket{111111101111100}, \ket{110111111110101},\\
\ket{111111101010111}, \ket{111100010111110}, \ket{110111111010001},\\
\ket{110101000000100}
\}$.


\bibliography{cite}

\begin{thebibliography}{40}%
\makeatletter
\providecommand \@ifxundefined [1]{%
 \@ifx{#1\undefined}
}%
\providecommand \@ifnum [1]{%
 \ifnum #1\expandafter \@firstoftwo
 \else \expandafter \@secondoftwo
 \fi
}%
\providecommand \@ifx [1]{%
 \ifx #1\expandafter \@firstoftwo
 \else \expandafter \@secondoftwo
 \fi
}%
\providecommand \natexlab [1]{#1}%
\providecommand \enquote  [1]{``#1''}%
\providecommand \bibnamefont  [1]{#1}%
\providecommand \bibfnamefont [1]{#1}%
\providecommand \citenamefont [1]{#1}%
\providecommand \href@noop [0]{\@secondoftwo}%
\providecommand \href [0]{\begingroup \@sanitize@url \@href}%
\providecommand \@href[1]{\@@startlink{#1}\@@href}%
\providecommand \@@href[1]{\endgroup#1\@@endlink}%
\providecommand \@sanitize@url [0]{\catcode `\\12\catcode `\$12\catcode
  `\&12\catcode `\#12\catcode `\^12\catcode `\_12\catcode `\%12\relax}%
\providecommand \@@startlink[1]{}%
\providecommand \@@endlink[0]{}%
\providecommand \url  [0]{\begingroup\@sanitize@url \@url }%
\providecommand \@url [1]{\endgroup\@href {#1}{\urlprefix }}%
\providecommand \urlprefix  [0]{URL }%
\providecommand \Eprint [0]{\href }%
\providecommand \doibase [0]{https://doi.org/}%
\providecommand \selectlanguage [0]{\@gobble}%
\providecommand \bibinfo  [0]{\@secondoftwo}%
\providecommand \bibfield  [0]{\@secondoftwo}%
\providecommand \translation [1]{[#1]}%
\providecommand \BibitemOpen [0]{}%
\providecommand \bibitemStop [0]{}%
\providecommand \bibitemNoStop [0]{.\EOS\space}%
\providecommand \EOS [0]{\spacefactor3000\relax}%
\providecommand \BibitemShut  [1]{\csname bibitem#1\endcsname}%
\let\auto@bib@innerbib\@empty
\bibitem [{\citenamefont {Noether}(1918)}]{Noether1918}%
  \BibitemOpen
  \bibfield  {author} {\bibinfo {author} {\bibfnamefont {E.}~\bibnamefont
  {Noether}},\ }\bibfield  {title} {\bibinfo {title} {Invariante
  variationsprobleme},\ }\href {http://eudml.org/doc/59024} {\bibfield
  {journal} {\bibinfo  {journal} {Nachrichten von der Gesellschaft der
  Wissenschaften zu Göttingen, Mathematisch-Physikalische Klasse}\ }\textbf
  {\bibinfo {volume} {1918}},\ \bibinfo {pages} {235} (\bibinfo {year}
  {1918})}\BibitemShut {NoStop}%
\bibitem [{\citenamefont {Gaillard}\ \emph {et~al.}(1999)\citenamefont
  {Gaillard}, \citenamefont {Grannis},\ and\ \citenamefont
  {Sciulli}}]{gaillard1999standard}%
  \BibitemOpen
  \bibfield  {author} {\bibinfo {author} {\bibfnamefont {M.~K.}\ \bibnamefont
  {Gaillard}}, \bibinfo {author} {\bibfnamefont {P.~D.}\ \bibnamefont
  {Grannis}},\ and\ \bibinfo {author} {\bibfnamefont {F.~J.}\ \bibnamefont
  {Sciulli}},\ }\bibfield  {title} {\bibinfo {title} {The standard model of
  particle physics},\ }\href@noop {} {\bibfield  {journal} {\bibinfo  {journal}
  {Reviews of Modern Physics}\ }\textbf {\bibinfo {volume} {71}},\ \bibinfo
  {pages} {S96} (\bibinfo {year} {1999})}\BibitemShut {NoStop}%
\bibitem [{\citenamefont {Kozlov}(1983)}]{kozlov1983integrability}%
  \BibitemOpen
  \bibfield  {author} {\bibinfo {author} {\bibfnamefont {V.~V.}\ \bibnamefont
  {Kozlov}},\ }\bibfield  {title} {\bibinfo {title} {Integrability and
  non-integrability in hamiltonian mechanics},\ }\href@noop {} {\bibfield
  {journal} {\bibinfo  {journal} {Russian Mathematical Surveys}\ }\textbf
  {\bibinfo {volume} {38}},\ \bibinfo {pages} {1} (\bibinfo {year}
  {1983})}\BibitemShut {NoStop}%
\bibitem [{\citenamefont {Landau}(1937)}]{Landau:1937obd}%
  \BibitemOpen
  \bibfield  {author} {\bibinfo {author} {\bibfnamefont {L.~D.}\ \bibnamefont
  {Landau}},\ }\bibfield  {title} {\bibinfo {title} {{On the theory of phase
  transitions}},\ }\href@noop {} {\bibfield  {journal} {\bibinfo  {journal}
  {Zh. Eksp. Teor. Fiz.}\ }\textbf {\bibinfo {volume} {7}},\ \bibinfo {pages}
  {19} (\bibinfo {year} {1937})}\BibitemShut {NoStop}%
\bibitem [{\citenamefont {Fowler}\ \emph {et~al.}(2012)\citenamefont {Fowler},
  \citenamefont {Mariantoni}, \citenamefont {Martinis},\ and\ \citenamefont
  {Cleland}}]{Fowler_2012}%
  \BibitemOpen
  \bibfield  {author} {\bibinfo {author} {\bibfnamefont {A.~G.}\ \bibnamefont
  {Fowler}}, \bibinfo {author} {\bibfnamefont {M.}~\bibnamefont {Mariantoni}},
  \bibinfo {author} {\bibfnamefont {J.~M.}\ \bibnamefont {Martinis}},\ and\
  \bibinfo {author} {\bibfnamefont {A.~N.}\ \bibnamefont {Cleland}},\
  }\bibfield  {title} {\bibinfo {title} {Surface codes: Towards practical
  large-scale quantum computation},\ }\bibfield  {journal} {\bibinfo  {journal}
  {Physical Review A}\ }\textbf {\bibinfo {volume} {86}},\ \href
  {https://doi.org/10.1103/physreva.86.032324} {10.1103/physreva.86.032324}
  (\bibinfo {year} {2012})\BibitemShut {NoStop}%
\bibitem [{\citenamefont {Bonet-Monroig}\ \emph {et~al.}(2018)\citenamefont
  {Bonet-Monroig}, \citenamefont {Sagastizabal}, \citenamefont {Singh},\ and\
  \citenamefont {O'Brien}}]{PhysRevA.98.062339}%
  \BibitemOpen
  \bibfield  {author} {\bibinfo {author} {\bibfnamefont {X.}~\bibnamefont
  {Bonet-Monroig}}, \bibinfo {author} {\bibfnamefont {R.}~\bibnamefont
  {Sagastizabal}}, \bibinfo {author} {\bibfnamefont {M.}~\bibnamefont
  {Singh}},\ and\ \bibinfo {author} {\bibfnamefont {T.~E.}\ \bibnamefont
  {O'Brien}},\ }\bibfield  {title} {\bibinfo {title} {Low-cost error mitigation
  by symmetry verification},\ }\href
  {https://doi.org/10.1103/PhysRevA.98.062339} {\bibfield  {journal} {\bibinfo
  {journal} {Phys. Rev. A}\ }\textbf {\bibinfo {volume} {98}},\ \bibinfo
  {pages} {062339} (\bibinfo {year} {2018})}\BibitemShut {NoStop}%
\bibitem [{\citenamefont {Sagastizabal}\ \emph {et~al.}(2019)\citenamefont
  {Sagastizabal}, \citenamefont {Bonet-Monroig}, \citenamefont {Singh},
  \citenamefont {Rol}, \citenamefont {Bultink}, \citenamefont {Fu},
  \citenamefont {Price}, \citenamefont {Ostroukh}, \citenamefont
  {Muthusubramanian}, \citenamefont {Bruno}, \citenamefont {Beekman},
  \citenamefont {Haider}, \citenamefont {O'Brien},\ and\ \citenamefont
  {DiCarlo}}]{PhysRevA.100.010302}%
  \BibitemOpen
  \bibfield  {author} {\bibinfo {author} {\bibfnamefont {R.}~\bibnamefont
  {Sagastizabal}}, \bibinfo {author} {\bibfnamefont {X.}~\bibnamefont
  {Bonet-Monroig}}, \bibinfo {author} {\bibfnamefont {M.}~\bibnamefont
  {Singh}}, \bibinfo {author} {\bibfnamefont {M.~A.}\ \bibnamefont {Rol}},
  \bibinfo {author} {\bibfnamefont {C.~C.}\ \bibnamefont {Bultink}}, \bibinfo
  {author} {\bibfnamefont {X.}~\bibnamefont {Fu}}, \bibinfo {author}
  {\bibfnamefont {C.~H.}\ \bibnamefont {Price}}, \bibinfo {author}
  {\bibfnamefont {V.~P.}\ \bibnamefont {Ostroukh}}, \bibinfo {author}
  {\bibfnamefont {N.}~\bibnamefont {Muthusubramanian}}, \bibinfo {author}
  {\bibfnamefont {A.}~\bibnamefont {Bruno}}, \bibinfo {author} {\bibfnamefont
  {M.}~\bibnamefont {Beekman}}, \bibinfo {author} {\bibfnamefont
  {N.}~\bibnamefont {Haider}}, \bibinfo {author} {\bibfnamefont {T.~E.}\
  \bibnamefont {O'Brien}},\ and\ \bibinfo {author} {\bibfnamefont
  {L.}~\bibnamefont {DiCarlo}},\ }\bibfield  {title} {\bibinfo {title}
  {Experimental error mitigation via symmetry verification in a variational
  quantum eigensolver},\ }\href {https://doi.org/10.1103/PhysRevA.100.010302}
  {\bibfield  {journal} {\bibinfo  {journal} {Phys. Rev. A}\ }\textbf {\bibinfo
  {volume} {100}},\ \bibinfo {pages} {010302} (\bibinfo {year}
  {2019})}\BibitemShut {NoStop}%
\bibitem [{\citenamefont {McClean}\ \emph {et~al.}(2017)\citenamefont
  {McClean}, \citenamefont {Kimchi-Schwartz}, \citenamefont {Carter},\ and\
  \citenamefont {de~Jong}}]{PhysRevA.95.042308}%
  \BibitemOpen
  \bibfield  {author} {\bibinfo {author} {\bibfnamefont {J.~R.}\ \bibnamefont
  {McClean}}, \bibinfo {author} {\bibfnamefont {M.~E.}\ \bibnamefont
  {Kimchi-Schwartz}}, \bibinfo {author} {\bibfnamefont {J.}~\bibnamefont
  {Carter}},\ and\ \bibinfo {author} {\bibfnamefont {W.~A.}\ \bibnamefont
  {de~Jong}},\ }\bibfield  {title} {\bibinfo {title} {Hybrid quantum-classical
  hierarchy for mitigation of decoherence and determination of excited
  states},\ }\href {https://doi.org/10.1103/PhysRevA.95.042308} {\bibfield
  {journal} {\bibinfo  {journal} {Phys. Rev. A}\ }\textbf {\bibinfo {volume}
  {95}},\ \bibinfo {pages} {042308} (\bibinfo {year} {2017})}\BibitemShut
  {NoStop}%
\bibitem [{\citenamefont {Cai}(2021)}]{Cai_2021}%
  \BibitemOpen
  \bibfield  {author} {\bibinfo {author} {\bibfnamefont {Z.}~\bibnamefont
  {Cai}},\ }\bibfield  {title} {\bibinfo {title} {Multi-exponential error
  extrapolation and combining error mitigation techniques for {NISQ}
  applications},\ }\bibfield  {journal} {\bibinfo  {journal} {npj Quantum
  Information}\ }\textbf {\bibinfo {volume} {7}},\ \href
  {https://doi.org/10.1038/s41534-021-00404-3} {10.1038/s41534-021-00404-3}
  (\bibinfo {year} {2021})\BibitemShut {NoStop}%
\bibitem [{\citenamefont {McClean}\ \emph {et~al.}(2020)\citenamefont
  {McClean}, \citenamefont {Jiang}, \citenamefont {Rubin}, \citenamefont
  {Babbush},\ and\ \citenamefont {Neven}}]{McClean2020}%
  \BibitemOpen
  \bibfield  {author} {\bibinfo {author} {\bibfnamefont {J.~R.}\ \bibnamefont
  {McClean}}, \bibinfo {author} {\bibfnamefont {Z.}~\bibnamefont {Jiang}},
  \bibinfo {author} {\bibfnamefont {N.~C.}\ \bibnamefont {Rubin}}, \bibinfo
  {author} {\bibfnamefont {R.}~\bibnamefont {Babbush}},\ and\ \bibinfo {author}
  {\bibfnamefont {H.}~\bibnamefont {Neven}},\ }\bibfield  {title} {\bibinfo
  {title} {Decoding quantum errors with subspace expansions},\ }\href
  {https://doi.org/10.1038/s41467-020-14341-w} {\bibfield  {journal} {\bibinfo
  {journal} {Nature Communications}\ }\textbf {\bibinfo {volume} {11}},\
  \bibinfo {pages} {636} (\bibinfo {year} {2020})}\BibitemShut {NoStop}%
\bibitem [{\citenamefont {Larocca}\ \emph {et~al.}(2022)\citenamefont
  {Larocca}, \citenamefont {Sauvage}, \citenamefont {Sbahi}, \citenamefont
  {Verdon}, \citenamefont {Coles},\ and\ \citenamefont
  {Cerezo}}]{PRXQuantum.3.030341}%
  \BibitemOpen
  \bibfield  {author} {\bibinfo {author} {\bibfnamefont {M.}~\bibnamefont
  {Larocca}}, \bibinfo {author} {\bibfnamefont {F.}~\bibnamefont {Sauvage}},
  \bibinfo {author} {\bibfnamefont {F.~M.}\ \bibnamefont {Sbahi}}, \bibinfo
  {author} {\bibfnamefont {G.}~\bibnamefont {Verdon}}, \bibinfo {author}
  {\bibfnamefont {P.~J.}\ \bibnamefont {Coles}},\ and\ \bibinfo {author}
  {\bibfnamefont {M.}~\bibnamefont {Cerezo}},\ }\bibfield  {title} {\bibinfo
  {title} {Group-invariant quantum machine learning},\ }\href
  {https://doi.org/10.1103/PRXQuantum.3.030341} {\bibfield  {journal} {\bibinfo
   {journal} {PRX Quantum}\ }\textbf {\bibinfo {volume} {3}},\ \bibinfo {pages}
  {030341} (\bibinfo {year} {2022})}\BibitemShut {NoStop}%
\bibitem [{\citenamefont {Daley}\ \emph {et~al.}(2022)\citenamefont {Daley},
  \citenamefont {Bloch}, \citenamefont {Kokail}, \citenamefont {Flannigan},
  \citenamefont {Pearson}, \citenamefont {Troyer},\ and\ \citenamefont
  {Zoller}}]{Daley2022-vl}%
  \BibitemOpen
  \bibfield  {author} {\bibinfo {author} {\bibfnamefont {A.~J.}\ \bibnamefont
  {Daley}}, \bibinfo {author} {\bibfnamefont {I.}~\bibnamefont {Bloch}},
  \bibinfo {author} {\bibfnamefont {C.}~\bibnamefont {Kokail}}, \bibinfo
  {author} {\bibfnamefont {S.}~\bibnamefont {Flannigan}}, \bibinfo {author}
  {\bibfnamefont {N.}~\bibnamefont {Pearson}}, \bibinfo {author} {\bibfnamefont
  {M.}~\bibnamefont {Troyer}},\ and\ \bibinfo {author} {\bibfnamefont
  {P.}~\bibnamefont {Zoller}},\ }\bibfield  {title} {\bibinfo {title}
  {Practical quantum advantage in quantum simulation},\ }\href@noop {}
  {\bibfield  {journal} {\bibinfo  {journal} {Nature}\ }\textbf {\bibinfo
  {volume} {607}},\ \bibinfo {pages} {667} (\bibinfo {year}
  {2022})}\BibitemShut {NoStop}%
\bibitem [{\citenamefont {Acharya}\ \emph {et~al.}(2022)\citenamefont
  {Acharya}, \citenamefont {Aleiner}, \citenamefont {Allen}, \citenamefont
  {Andersen}, \citenamefont {Ansmann}, \citenamefont {Arute}, \citenamefont
  {Arya}, \citenamefont {Asfaw}, \citenamefont {Atalaya}, \citenamefont
  {Babbush} \emph {et~al.}}]{acharya2022suppressing}%
  \BibitemOpen
  \bibfield  {author} {\bibinfo {author} {\bibfnamefont {R.}~\bibnamefont
  {Acharya}}, \bibinfo {author} {\bibfnamefont {I.}~\bibnamefont {Aleiner}},
  \bibinfo {author} {\bibfnamefont {R.}~\bibnamefont {Allen}}, \bibinfo
  {author} {\bibfnamefont {T.~I.}\ \bibnamefont {Andersen}}, \bibinfo {author}
  {\bibfnamefont {M.}~\bibnamefont {Ansmann}}, \bibinfo {author} {\bibfnamefont
  {F.}~\bibnamefont {Arute}}, \bibinfo {author} {\bibfnamefont
  {K.}~\bibnamefont {Arya}}, \bibinfo {author} {\bibfnamefont {A.}~\bibnamefont
  {Asfaw}}, \bibinfo {author} {\bibfnamefont {J.}~\bibnamefont {Atalaya}},
  \bibinfo {author} {\bibfnamefont {R.}~\bibnamefont {Babbush}}, \emph
  {et~al.},\ }\bibfield  {title} {\bibinfo {title} {Suppressing quantum errors
  by scaling a surface code logical qubit},\ }\href@noop {} {\bibfield
  {journal} {\bibinfo  {journal} {arXiv preprint arXiv:2207.06431}\ } (\bibinfo
  {year} {2022})}\BibitemShut {NoStop}%
\bibitem [{\citenamefont {Cai}\ \emph {et~al.}(2022)\citenamefont {Cai},
  \citenamefont {Babbush}, \citenamefont {Benjamin}, \citenamefont {Endo},
  \citenamefont {Huggins}, \citenamefont {Li}, \citenamefont {McClean},\ and\
  \citenamefont {O'Brien}}]{cai2022quantum}%
  \BibitemOpen
  \bibfield  {author} {\bibinfo {author} {\bibfnamefont {Z.}~\bibnamefont
  {Cai}}, \bibinfo {author} {\bibfnamefont {R.}~\bibnamefont {Babbush}},
  \bibinfo {author} {\bibfnamefont {S.~C.}\ \bibnamefont {Benjamin}}, \bibinfo
  {author} {\bibfnamefont {S.}~\bibnamefont {Endo}}, \bibinfo {author}
  {\bibfnamefont {W.~J.}\ \bibnamefont {Huggins}}, \bibinfo {author}
  {\bibfnamefont {Y.}~\bibnamefont {Li}}, \bibinfo {author} {\bibfnamefont
  {J.~R.}\ \bibnamefont {McClean}},\ and\ \bibinfo {author} {\bibfnamefont
  {T.~E.}\ \bibnamefont {O'Brien}},\ }\bibfield  {title} {\bibinfo {title}
  {Quantum error mitigation},\ }\href@noop {} {\bibfield  {journal} {\bibinfo
  {journal} {arXiv preprint arXiv:2210.00921}\ } (\bibinfo {year}
  {2022})}\BibitemShut {NoStop}%
\bibitem [{\citenamefont {Quek}\ \emph {et~al.}(2022)\citenamefont {Quek},
  \citenamefont {Fran{\c{c}}a}, \citenamefont {Khatri}, \citenamefont {Meyer},\
  and\ \citenamefont {Eisert}}]{quek2022exponentially}%
  \BibitemOpen
  \bibfield  {author} {\bibinfo {author} {\bibfnamefont {Y.}~\bibnamefont
  {Quek}}, \bibinfo {author} {\bibfnamefont {D.~S.}\ \bibnamefont
  {Fran{\c{c}}a}}, \bibinfo {author} {\bibfnamefont {S.}~\bibnamefont
  {Khatri}}, \bibinfo {author} {\bibfnamefont {J.~J.}\ \bibnamefont {Meyer}},\
  and\ \bibinfo {author} {\bibfnamefont {J.}~\bibnamefont {Eisert}},\
  }\bibfield  {title} {\bibinfo {title} {Exponentially tighter bounds on
  limitations of quantum error mitigation},\ }\href@noop {} {\bibfield
  {journal} {\bibinfo  {journal} {arXiv preprint arXiv:2210.11505}\ } (\bibinfo
  {year} {2022})}\BibitemShut {NoStop}%
\bibitem [{\citenamefont {Cao}\ \emph {et~al.}(2019)\citenamefont {Cao},
  \citenamefont {Romero}, \citenamefont {Olson}, \citenamefont {Degroote},
  \citenamefont {Johnson}, \citenamefont {Kieferová}, \citenamefont
  {Kivlichan}, \citenamefont {Menke}, \citenamefont {Peropadre}, \citenamefont
  {Sawaya} \emph {et~al.}}]{doi:10.1021/acs.chemrev.8b00803}%
  \BibitemOpen
  \bibfield  {author} {\bibinfo {author} {\bibfnamefont {Y.}~\bibnamefont
  {Cao}}, \bibinfo {author} {\bibfnamefont {J.}~\bibnamefont {Romero}},
  \bibinfo {author} {\bibfnamefont {J.~P.}\ \bibnamefont {Olson}}, \bibinfo
  {author} {\bibfnamefont {M.}~\bibnamefont {Degroote}}, \bibinfo {author}
  {\bibfnamefont {P.~D.}\ \bibnamefont {Johnson}}, \bibinfo {author}
  {\bibfnamefont {M.}~\bibnamefont {Kieferová}}, \bibinfo {author}
  {\bibfnamefont {I.~D.}\ \bibnamefont {Kivlichan}}, \bibinfo {author}
  {\bibfnamefont {T.}~\bibnamefont {Menke}}, \bibinfo {author} {\bibfnamefont
  {B.}~\bibnamefont {Peropadre}}, \bibinfo {author} {\bibfnamefont {N.~P.~D.}\
  \bibnamefont {Sawaya}}, \emph {et~al.},\ }\bibfield  {title} {\bibinfo
  {title} {Quantum chemistry in the age of quantum computing},\ }\href
  {https://doi.org/10.1021/acs.chemrev.8b00803} {\bibfield  {journal} {\bibinfo
   {journal} {Chemical Reviews}\ }\textbf {\bibinfo {volume} {119}},\ \bibinfo
  {pages} {10856} (\bibinfo {year} {2019})},\ \bibinfo {note} {pMID:
  31469277},\ \Eprint
  {https://arxiv.org/abs/https://doi.org/10.1021/acs.chemrev.8b00803}
  {https://doi.org/10.1021/acs.chemrev.8b00803} \BibitemShut {NoStop}%
\bibitem [{\citenamefont {Takagi}\ \emph {et~al.}(2022)\citenamefont {Takagi},
  \citenamefont {Endo}, \citenamefont {Minagawa},\ and\ \citenamefont
  {Gu}}]{Takagi_2022}%
  \BibitemOpen
  \bibfield  {author} {\bibinfo {author} {\bibfnamefont {R.}~\bibnamefont
  {Takagi}}, \bibinfo {author} {\bibfnamefont {S.}~\bibnamefont {Endo}},
  \bibinfo {author} {\bibfnamefont {S.}~\bibnamefont {Minagawa}},\ and\
  \bibinfo {author} {\bibfnamefont {M.}~\bibnamefont {Gu}},\ }\bibfield
  {title} {\bibinfo {title} {Fundamental limits of quantum error mitigation},\
  }\bibfield  {journal} {\bibinfo  {journal} {npj Quantum Information}\
  }\textbf {\bibinfo {volume} {8}},\ \href
  {https://doi.org/10.1038/s41534-022-00618-z} {10.1038/s41534-022-00618-z}
  (\bibinfo {year} {2022})\BibitemShut {NoStop}%
\bibitem [{\citenamefont {Arute}\ \emph
  {et~al.}(2020{\natexlab{a}})\citenamefont {Arute}, \citenamefont {Arya},
  \citenamefont {Babbush}, \citenamefont {Bacon}, \citenamefont {Bardin},
  \citenamefont {Barends}, \citenamefont {Bengtsson}, \citenamefont {Boixo},
  \citenamefont {Broughton}, \citenamefont {Buckley} \emph
  {et~al.}}]{arute2020observation}%
  \BibitemOpen
  \bibfield  {author} {\bibinfo {author} {\bibfnamefont {F.}~\bibnamefont
  {Arute}}, \bibinfo {author} {\bibfnamefont {K.}~\bibnamefont {Arya}},
  \bibinfo {author} {\bibfnamefont {R.}~\bibnamefont {Babbush}}, \bibinfo
  {author} {\bibfnamefont {D.}~\bibnamefont {Bacon}}, \bibinfo {author}
  {\bibfnamefont {J.~C.}\ \bibnamefont {Bardin}}, \bibinfo {author}
  {\bibfnamefont {R.}~\bibnamefont {Barends}}, \bibinfo {author} {\bibfnamefont
  {A.}~\bibnamefont {Bengtsson}}, \bibinfo {author} {\bibfnamefont
  {S.}~\bibnamefont {Boixo}}, \bibinfo {author} {\bibfnamefont
  {M.}~\bibnamefont {Broughton}}, \bibinfo {author} {\bibfnamefont {B.~B.}\
  \bibnamefont {Buckley}}, \emph {et~al.},\ }\bibfield  {title} {\bibinfo
  {title} {Observation of separated dynamics of charge and spin in the
  fermi-hubbard model},\ }\href@noop {} {\bibfield  {journal} {\bibinfo
  {journal} {arXiv preprint arXiv:2010.07965}\ } (\bibinfo {year}
  {2020}{\natexlab{a}})}\BibitemShut {NoStop}%
\bibitem [{\citenamefont {Arute}\ \emph
  {et~al.}(2020{\natexlab{b}})\citenamefont {Arute}, \citenamefont {Arya},
  \citenamefont {Babbush}, \citenamefont {Bacon}, \citenamefont {Bardin},
  \citenamefont {Barends}, \citenamefont {Boixo}, \citenamefont {Broughton},
  \citenamefont {Buckley}, \citenamefont {Buell} \emph
  {et~al.}}]{doi:10.1126/science.abb9811}%
  \BibitemOpen
  \bibfield  {author} {\bibinfo {author} {\bibfnamefont {F.}~\bibnamefont
  {Arute}}, \bibinfo {author} {\bibfnamefont {K.}~\bibnamefont {Arya}},
  \bibinfo {author} {\bibfnamefont {R.}~\bibnamefont {Babbush}}, \bibinfo
  {author} {\bibfnamefont {D.}~\bibnamefont {Bacon}}, \bibinfo {author}
  {\bibfnamefont {J.~C.}\ \bibnamefont {Bardin}}, \bibinfo {author}
  {\bibfnamefont {R.}~\bibnamefont {Barends}}, \bibinfo {author} {\bibfnamefont
  {S.}~\bibnamefont {Boixo}}, \bibinfo {author} {\bibfnamefont
  {M.}~\bibnamefont {Broughton}}, \bibinfo {author} {\bibfnamefont {B.~B.}\
  \bibnamefont {Buckley}}, \bibinfo {author} {\bibfnamefont {D.~A.}\
  \bibnamefont {Buell}}, \emph {et~al.},\ }\bibfield  {title} {\bibinfo {title}
  {Hartree-fock on a superconducting qubit quantum computer},\ }\href
  {https://doi.org/10.1126/science.abb9811} {\bibfield  {journal} {\bibinfo
  {journal} {Science}\ }\textbf {\bibinfo {volume} {369}},\ \bibinfo {pages}
  {1084} (\bibinfo {year} {2020}{\natexlab{b}})}\BibitemShut {NoStop}%
\bibitem [{\citenamefont {Jones}\ \emph {et~al.}(2022)\citenamefont {Jones},
  \citenamefont {Hillberry}, \citenamefont {Jones}, \citenamefont {Fasihi},
  \citenamefont {Roushan}, \citenamefont {Jiang}, \citenamefont {Ho},
  \citenamefont {Neill}, \citenamefont {Ostby}, \citenamefont {Graf},
  \citenamefont {Kapit},\ and\ \citenamefont {Carr}}]{Jones2022}%
  \BibitemOpen
  \bibfield  {author} {\bibinfo {author} {\bibfnamefont {E.~B.}\ \bibnamefont
  {Jones}}, \bibinfo {author} {\bibfnamefont {L.~E.}\ \bibnamefont
  {Hillberry}}, \bibinfo {author} {\bibfnamefont {M.~T.}\ \bibnamefont
  {Jones}}, \bibinfo {author} {\bibfnamefont {M.}~\bibnamefont {Fasihi}},
  \bibinfo {author} {\bibfnamefont {P.}~\bibnamefont {Roushan}}, \bibinfo
  {author} {\bibfnamefont {Z.}~\bibnamefont {Jiang}}, \bibinfo {author}
  {\bibfnamefont {A.}~\bibnamefont {Ho}}, \bibinfo {author} {\bibfnamefont
  {C.}~\bibnamefont {Neill}}, \bibinfo {author} {\bibfnamefont
  {E.}~\bibnamefont {Ostby}}, \bibinfo {author} {\bibfnamefont
  {P.}~\bibnamefont {Graf}}, \bibinfo {author} {\bibfnamefont {E.}~\bibnamefont
  {Kapit}},\ and\ \bibinfo {author} {\bibfnamefont {L.~D.}\ \bibnamefont
  {Carr}},\ }\bibfield  {title} {\bibinfo {title} {Small-world complex network
  generation on a digital quantum processor},\ }\href
  {https://doi.org/10.1038/s41467-022-32056-y} {\bibfield  {journal} {\bibinfo
  {journal} {Nature Communications}\ }\textbf {\bibinfo {volume} {13}},\
  \bibinfo {pages} {4483} (\bibinfo {year} {2022})}\BibitemShut {NoStop}%
\bibitem [{\citenamefont {H{\'e}bert}\ \emph {et~al.}(2001)\citenamefont
  {H{\'e}bert}, \citenamefont {Batrouni}, \citenamefont {Scalettar},
  \citenamefont {Schmid}, \citenamefont {Troyer},\ and\ \citenamefont
  {Dorneich}}]{hebert2001quantum}%
  \BibitemOpen
  \bibfield  {author} {\bibinfo {author} {\bibfnamefont {F.}~\bibnamefont
  {H{\'e}bert}}, \bibinfo {author} {\bibfnamefont {G.~G.}\ \bibnamefont
  {Batrouni}}, \bibinfo {author} {\bibfnamefont {R.~T.}\ \bibnamefont
  {Scalettar}}, \bibinfo {author} {\bibfnamefont {G.}~\bibnamefont {Schmid}},
  \bibinfo {author} {\bibfnamefont {M.}~\bibnamefont {Troyer}},\ and\ \bibinfo
  {author} {\bibfnamefont {A.}~\bibnamefont {Dorneich}},\ }\bibfield  {title}
  {\bibinfo {title} {Quantum phase transitions in the two-dimensional hardcore
  boson model},\ }\href@noop {} {\bibfield  {journal} {\bibinfo  {journal}
  {Physical Review B}\ }\textbf {\bibinfo {volume} {65}},\ \bibinfo {pages}
  {014513} (\bibinfo {year} {2001})}\BibitemShut {NoStop}%
\bibitem [{\citenamefont {Anschuetz}\ \emph {et~al.}(2022)\citenamefont
  {Anschuetz}, \citenamefont {Bauer}, \citenamefont {Kiani},\ and\
  \citenamefont {Lloyd}}]{anschuetz2022efficient}%
  \BibitemOpen
  \bibfield  {author} {\bibinfo {author} {\bibfnamefont {E.~R.}\ \bibnamefont
  {Anschuetz}}, \bibinfo {author} {\bibfnamefont {A.}~\bibnamefont {Bauer}},
  \bibinfo {author} {\bibfnamefont {B.~T.}\ \bibnamefont {Kiani}},\ and\
  \bibinfo {author} {\bibfnamefont {S.}~\bibnamefont {Lloyd}},\ }\bibfield
  {title} {\bibinfo {title} {Efficient classical algorithms for simulating
  symmetric quantum systems},\ }\href@noop {} {\bibfield  {journal} {\bibinfo
  {journal} {arXiv preprint arXiv:2211.16998}\ } (\bibinfo {year}
  {2022})}\BibitemShut {NoStop}%
\bibitem [{\citenamefont {Radjavi}\ and\ \citenamefont
  {Rosenthal}(2003)}]{radjavi2003invariant}%
  \BibitemOpen
  \bibfield  {author} {\bibinfo {author} {\bibfnamefont {H.}~\bibnamefont
  {Radjavi}}\ and\ \bibinfo {author} {\bibfnamefont {P.}~\bibnamefont
  {Rosenthal}},\ }\href@noop {} {\emph {\bibinfo {title} {Invariant
  subspaces}}}\ (\bibinfo  {publisher} {Courier Corporation},\ \bibinfo {year}
  {2003})\BibitemShut {NoStop}%
\bibitem [{\citenamefont {Suzuki}(1991)}]{Suzuki1991}%
  \BibitemOpen
  \bibfield  {author} {\bibinfo {author} {\bibfnamefont {M.}~\bibnamefont
  {Suzuki}},\ }\bibfield  {title} {\bibinfo {title} {General theory of fractal
  path integrals with applications to many-body theories and statistical
  physics},\ }\href {https://doi.org/10.1063/1.529425} {\bibfield  {journal}
  {\bibinfo  {journal} {Journal of Mathematical Physics}\ }\textbf {\bibinfo
  {volume} {32}},\ \bibinfo {pages} {400} (\bibinfo {year} {1991})}\BibitemShut
  {NoStop}%
\bibitem [{\citenamefont {Hillberry}\ \emph {et~al.}(2021)\citenamefont
  {Hillberry}, \citenamefont {Jones}, \citenamefont {Vargas}, \citenamefont
  {Rall}, \citenamefont {Halpern}, \citenamefont {Bao}, \citenamefont
  {Notarnicola}, \citenamefont {Montangero},\ and\ \citenamefont
  {Carr}}]{Hillberry2021}%
  \BibitemOpen
  \bibfield  {author} {\bibinfo {author} {\bibfnamefont {L.~E.}\ \bibnamefont
  {Hillberry}}, \bibinfo {author} {\bibfnamefont {M.~T.}\ \bibnamefont
  {Jones}}, \bibinfo {author} {\bibfnamefont {D.~L.}\ \bibnamefont {Vargas}},
  \bibinfo {author} {\bibfnamefont {P.}~\bibnamefont {Rall}}, \bibinfo {author}
  {\bibfnamefont {N.~Y.}\ \bibnamefont {Halpern}}, \bibinfo {author}
  {\bibfnamefont {N.}~\bibnamefont {Bao}}, \bibinfo {author} {\bibfnamefont
  {S.}~\bibnamefont {Notarnicola}}, \bibinfo {author} {\bibfnamefont
  {S.}~\bibnamefont {Montangero}},\ and\ \bibinfo {author} {\bibfnamefont
  {L.~D.}\ \bibnamefont {Carr}},\ }\bibfield  {title} {\bibinfo {title}
  {Entangled quantum cellular automata, physical complexity, and goldilocks
  rules},\ }\href {https://doi.org/10.1088/2058-9565/ac1c41} {\bibfield
  {journal} {\bibinfo  {journal} {Quantum Science and Technology}\ }\textbf
  {\bibinfo {volume} {6}},\ \bibinfo {pages} {045017} (\bibinfo {year}
  {2021})}\BibitemShut {NoStop}%
\bibitem [{\citenamefont {Lidl}(1998)}]{lidl}%
  \BibitemOpen
  \bibfield  {author} {\bibinfo {author} {\bibfnamefont {R.}~\bibnamefont
  {Lidl}},\ }\href@noop {} {\emph {\bibinfo {title} {Applied abstract
  algebra}}},\ \bibinfo {edition} {second edition.}\ ed.,\ Undergraduate Texts
  in Mathematics\ (\bibinfo  {publisher} {Springer},\ \bibinfo {address} {New
  York},\ \bibinfo {year} {1998})\ p.\ \bibinfo {pages} {337}\BibitemShut
  {NoStop}%
\bibitem [{\citenamefont {Fischer}\ and\ \citenamefont
  {Meyer}(1971)}]{fischer1971}%
  \BibitemOpen
  \bibfield  {author} {\bibinfo {author} {\bibfnamefont {M.~J.}\ \bibnamefont
  {Fischer}}\ and\ \bibinfo {author} {\bibfnamefont {A.~R.}\ \bibnamefont
  {Meyer}},\ }\bibfield  {title} {\bibinfo {title} {Boolean matrix
  multiplication and transitive closure},\ }in\ \href
  {https://doi.org/10.1109/SWAT.1971.4} {\emph {\bibinfo {booktitle} {12th
  Annual Symposium on Switching and Automata Theory (swat 1971)}}}\ (\bibinfo
  {year} {1971})\ pp.\ \bibinfo {pages} {129--131}\BibitemShut {NoStop}%
\bibitem [{\citenamefont {Marvian}(2020)}]{marvian2020locality}%
  \BibitemOpen
  \bibfield  {author} {\bibinfo {author} {\bibfnamefont {I.}~\bibnamefont
  {Marvian}},\ }\bibfield  {title} {\bibinfo {title} {Locality and conservation
  laws: How, in the presence of symmetry, locality restricts realizable
  unitaries},\ }\href@noop {} {\bibfield  {journal} {\bibinfo  {journal} {arXiv
  preprint arXiv:2003.05524}\ } (\bibinfo {year} {2020})}\BibitemShut {NoStop}%
\bibitem [{\citenamefont {Skiena}(2008)}]{Skiena2008}%
  \BibitemOpen
  \bibfield  {author} {\bibinfo {author} {\bibfnamefont {S.~S.}\ \bibnamefont
  {Skiena}},\ }\bibinfo {title} {Sorting and searching},\ in\ \href
  {https://doi.org/10.1007/978-1-84800-070-4_4} {\emph {\bibinfo {booktitle}
  {The Algorithm Design Manual}}}\ (\bibinfo  {publisher} {Springer London},\
  \bibinfo {address} {London},\ \bibinfo {year} {2008})\ p.\ \bibinfo {pages}
  {477}\BibitemShut {NoStop}%
\bibitem [{\citenamefont {Lee}(1961)}]{5219222}%
  \BibitemOpen
  \bibfield  {author} {\bibinfo {author} {\bibfnamefont {C.~Y.}\ \bibnamefont
  {Lee}},\ }\bibfield  {title} {\bibinfo {title} {An algorithm for path
  connections and its applications},\ }\href
  {https://doi.org/10.1109/TEC.1961.5219222} {\bibfield  {journal} {\bibinfo
  {journal} {IRE Transactions on Electronic Computers}\ }\textbf {\bibinfo
  {volume} {EC-10}},\ \bibinfo {pages} {346} (\bibinfo {year}
  {1961})}\BibitemShut {NoStop}%
\bibitem [{\citenamefont {Kullback}\ and\ \citenamefont
  {Leibler}(1951)}]{kullback1951information}%
  \BibitemOpen
  \bibfield  {author} {\bibinfo {author} {\bibfnamefont {S.}~\bibnamefont
  {Kullback}}\ and\ \bibinfo {author} {\bibfnamefont {R.~A.}\ \bibnamefont
  {Leibler}},\ }\bibfield  {title} {\bibinfo {title} {On information and
  sufficiency},\ }\href@noop {} {\bibfield  {journal} {\bibinfo  {journal} {The
  annals of mathematical statistics}\ }\textbf {\bibinfo {volume} {22}},\
  \bibinfo {pages} {79} (\bibinfo {year} {1951})}\BibitemShut {NoStop}%
\bibitem [{\citenamefont {Kapit}\ \emph {et~al.}(2020)\citenamefont {Kapit},
  \citenamefont {Roushan}, \citenamefont {Neill}, \citenamefont {Boixo},\ and\
  \citenamefont {Smelyanskiy}}]{PhysRevResearch.2.043042}%
  \BibitemOpen
  \bibfield  {author} {\bibinfo {author} {\bibfnamefont {E.}~\bibnamefont
  {Kapit}}, \bibinfo {author} {\bibfnamefont {P.}~\bibnamefont {Roushan}},
  \bibinfo {author} {\bibfnamefont {C.}~\bibnamefont {Neill}}, \bibinfo
  {author} {\bibfnamefont {S.}~\bibnamefont {Boixo}},\ and\ \bibinfo {author}
  {\bibfnamefont {V.}~\bibnamefont {Smelyanskiy}},\ }\bibfield  {title}
  {\bibinfo {title} {Entanglement and complexity of interacting qubits subject
  to asymmetric noise},\ }\href
  {https://doi.org/10.1103/PhysRevResearch.2.043042} {\bibfield  {journal}
  {\bibinfo  {journal} {Phys. Rev. Res.}\ }\textbf {\bibinfo {volume} {2}},\
  \bibinfo {pages} {043042} (\bibinfo {year} {2020})}\BibitemShut {NoStop}%
\bibitem [{\citenamefont {Thingna}\ and\ \citenamefont
  {Manzano}(2021)}]{thingna2021degenerated}%
  \BibitemOpen
  \bibfield  {author} {\bibinfo {author} {\bibfnamefont {J.}~\bibnamefont
  {Thingna}}\ and\ \bibinfo {author} {\bibfnamefont {D.}~\bibnamefont
  {Manzano}},\ }\bibfield  {title} {\bibinfo {title} {Degenerated liouvillians
  and steady-state reduced density matrices},\ }\href@noop {} {\bibfield
  {journal} {\bibinfo  {journal} {Chaos: An Interdisciplinary Journal of
  Nonlinear Science}\ }\textbf {\bibinfo {volume} {31}},\ \bibinfo {pages}
  {073114} (\bibinfo {year} {2021})}\BibitemShut {NoStop}%
\bibitem [{\citenamefont {Endo}\ \emph {et~al.}(2021)\citenamefont {Endo},
  \citenamefont {Cai}, \citenamefont {Benjamin},\ and\ \citenamefont
  {Yuan}}]{endo2021hybrid}%
  \BibitemOpen
  \bibfield  {author} {\bibinfo {author} {\bibfnamefont {S.}~\bibnamefont
  {Endo}}, \bibinfo {author} {\bibfnamefont {Z.}~\bibnamefont {Cai}}, \bibinfo
  {author} {\bibfnamefont {S.~C.}\ \bibnamefont {Benjamin}},\ and\ \bibinfo
  {author} {\bibfnamefont {X.}~\bibnamefont {Yuan}},\ }\bibfield  {title}
  {\bibinfo {title} {Hybrid quantum-classical algorithms and quantum error
  mitigation},\ }\href@noop {} {\bibfield  {journal} {\bibinfo  {journal}
  {Journal of the Physical Society of Japan}\ }\textbf {\bibinfo {volume}
  {90}},\ \bibinfo {pages} {032001} (\bibinfo {year} {2021})}\BibitemShut
  {NoStop}%
\bibitem [{\citenamefont {Needleman}\ and\ \citenamefont
  {Wunsch}(1970)}]{NEEDLEMAN1970443}%
  \BibitemOpen
  \bibfield  {author} {\bibinfo {author} {\bibfnamefont {S.~B.}\ \bibnamefont
  {Needleman}}\ and\ \bibinfo {author} {\bibfnamefont {C.~D.}\ \bibnamefont
  {Wunsch}},\ }\bibfield  {title} {\bibinfo {title} {A general method
  applicable to the search for similarities in the amino acid sequence of two
  proteins},\ }\href
  {https://doi.org/https://doi.org/10.1016/0022-2836(70)90057-4} {\bibfield
  {journal} {\bibinfo  {journal} {Journal of Molecular Biology}\ }\textbf
  {\bibinfo {volume} {48}},\ \bibinfo {pages} {443} (\bibinfo {year}
  {1970})}\BibitemShut {NoStop}%
\bibitem [{\citenamefont {Yunger~Halpern}\ \emph {et~al.}(2016)\citenamefont
  {Yunger~Halpern}, \citenamefont {Faist}, \citenamefont {Oppenheim},\ and\
  \citenamefont {Winter}}]{yunger2016microcanonical}%
  \BibitemOpen
  \bibfield  {author} {\bibinfo {author} {\bibfnamefont {N.}~\bibnamefont
  {Yunger~Halpern}}, \bibinfo {author} {\bibfnamefont {P.}~\bibnamefont
  {Faist}}, \bibinfo {author} {\bibfnamefont {J.}~\bibnamefont {Oppenheim}},\
  and\ \bibinfo {author} {\bibfnamefont {A.}~\bibnamefont {Winter}},\
  }\bibfield  {title} {\bibinfo {title} {Microcanonical and resource-theoretic
  derivations of the thermal state of a quantum system with noncommuting
  charges},\ }\href@noop {} {\bibfield  {journal} {\bibinfo  {journal} {Nature
  communications}\ }\textbf {\bibinfo {volume} {7}},\ \bibinfo {pages} {1}
  (\bibinfo {year} {2016})}\BibitemShut {NoStop}%
\bibitem [{\citenamefont {Majidy}\ \emph {et~al.}(2023)\citenamefont {Majidy},
  \citenamefont {Lasek}, \citenamefont {Huse},\ and\ \citenamefont
  {Halpern}}]{majidy2023non}%
  \BibitemOpen
  \bibfield  {author} {\bibinfo {author} {\bibfnamefont {S.}~\bibnamefont
  {Majidy}}, \bibinfo {author} {\bibfnamefont {A.}~\bibnamefont {Lasek}},
  \bibinfo {author} {\bibfnamefont {D.~A.}\ \bibnamefont {Huse}},\ and\
  \bibinfo {author} {\bibfnamefont {N.~Y.}\ \bibnamefont {Halpern}},\
  }\bibfield  {title} {\bibinfo {title} {Non-abelian symmetry can increase
  entanglement entropy},\ }\href@noop {} {\bibfield  {journal} {\bibinfo
  {journal} {Physical Review B}\ }\textbf {\bibinfo {volume} {107}},\ \bibinfo
  {pages} {045102} (\bibinfo {year} {2023})}\BibitemShut {NoStop}%
\bibitem [{\citenamefont {Murthy}\ \emph {et~al.}(2022)\citenamefont {Murthy},
  \citenamefont {Babakhani}, \citenamefont {Iniguez}, \citenamefont
  {Srednicki},\ and\ \citenamefont {Halpern}}]{murthy2022non}%
  \BibitemOpen
  \bibfield  {author} {\bibinfo {author} {\bibfnamefont {C.}~\bibnamefont
  {Murthy}}, \bibinfo {author} {\bibfnamefont {A.}~\bibnamefont {Babakhani}},
  \bibinfo {author} {\bibfnamefont {F.}~\bibnamefont {Iniguez}}, \bibinfo
  {author} {\bibfnamefont {M.}~\bibnamefont {Srednicki}},\ and\ \bibinfo
  {author} {\bibfnamefont {N.~Y.}\ \bibnamefont {Halpern}},\ }\bibfield
  {title} {\bibinfo {title} {Non-abelian eigenstate thermalization
  hypothesis},\ }\href@noop {} {\bibfield  {journal} {\bibinfo  {journal}
  {arXiv preprint arXiv:2206.05310}\ } (\bibinfo {year} {2022})}\BibitemShut
  {NoStop}%
\bibitem [{\citenamefont {Halpern}\ \emph {et~al.}(2020)\citenamefont
  {Halpern}, \citenamefont {Beverland},\ and\ \citenamefont
  {Kalev}}]{halpern2020noncommuting}%
  \BibitemOpen
  \bibfield  {author} {\bibinfo {author} {\bibfnamefont {N.~Y.}\ \bibnamefont
  {Halpern}}, \bibinfo {author} {\bibfnamefont {M.~E.}\ \bibnamefont
  {Beverland}},\ and\ \bibinfo {author} {\bibfnamefont {A.}~\bibnamefont
  {Kalev}},\ }\bibfield  {title} {\bibinfo {title} {Noncommuting conserved
  charges in quantum many-body thermalization},\ }\href@noop {} {\bibfield
  {journal} {\bibinfo  {journal} {Physical Review E}\ }\textbf {\bibinfo
  {volume} {101}},\ \bibinfo {pages} {042117} (\bibinfo {year}
  {2020})}\BibitemShut {NoStop}%
\bibitem [{\citenamefont {Kranzl}\ \emph {et~al.}(2022)\citenamefont {Kranzl},
  \citenamefont {Lasek}, \citenamefont {Joshi}, \citenamefont {Kalev},
  \citenamefont {Blatt}, \citenamefont {Roos},\ and\ \citenamefont
  {Halpern}}]{kranzl2022experimental}%
  \BibitemOpen
  \bibfield  {author} {\bibinfo {author} {\bibfnamefont {F.}~\bibnamefont
  {Kranzl}}, \bibinfo {author} {\bibfnamefont {A.}~\bibnamefont {Lasek}},
  \bibinfo {author} {\bibfnamefont {M.~K.}\ \bibnamefont {Joshi}}, \bibinfo
  {author} {\bibfnamefont {A.}~\bibnamefont {Kalev}}, \bibinfo {author}
  {\bibfnamefont {R.}~\bibnamefont {Blatt}}, \bibinfo {author} {\bibfnamefont
  {C.~F.}\ \bibnamefont {Roos}},\ and\ \bibinfo {author} {\bibfnamefont
  {N.~Y.}\ \bibnamefont {Halpern}},\ }\bibfield  {title} {\bibinfo {title}
  {Experimental observation of thermalisation with noncommuting charges},\
  }\href@noop {} {\bibfield  {journal} {\bibinfo  {journal} {arXiv preprint
  arXiv:2202.04652}\ } (\bibinfo {year} {2022})}\BibitemShut {NoStop}%
\end{thebibliography}%

\end{document}